\documentclass[manuscript, screen, authorversion, nonacm, preprint]{jair}

\setcopyright{cc}
\setcopyright{none}
\copyrightyear{2025}
\acmDOI{10.1613/jair.1.xxxxx}

\JAIRAE{Insert JAIR AE Name}
\JAIRTrack{} %
\acmVolume{4}
\acmArticle{6}
\acmMonth{8}
\acmYear{2025}

\RequirePackage[
  datamodel=acmdatamodel,
  style=acmauthoryear,
  backend=biber,
  giveninits=true,
  uniquename=init
  ]{biblatex}

\usepackage{amsmath}
\usepackage{amsfonts}
\usepackage{amsthm}
\usepackage{cleveref}
\usepackage{xspace}
\usepackage{xcolor}
\usepackage{tikz}
\usetikzlibrary{backgrounds}
\usepackage{nicefrac}
\usepackage{bm}

\usepackage{tabularx}
\usepackage{booktabs}
\usepackage{multirow}

\newcommand{\N}{\mathbb{N}}

\newcommand{\Yes}{\textsc{Yes}\xspace}
\newcommand{\No}{\textsc{No}\xspace}

\newcommand{\agents}{{\ensuremath{{N}}}}    %
\newcommand{\numAgents}{{\ensuremath{{n}}}} %
\newcommand{\wFn}{w}

\newcommand{\items}{\ensuremath{{M}}}       %
\newcommand{\numItems}{\ensuremath{{m}}}    %
\newcommand{\val}{{v}}                      %
\newcommand{\si}{{s}}                      %
\newcommand{\alloc}{{\mathcal{A}}}          %
\newcommand{\bundle}{{A}}
\newcommand{\SI}{\operatorname{SI}}

\newcommand{\probName}[1]{\textsc{#1}\xspace}

\newcommand{\FDSI}[1][$\mathcal{F}$]{\probName{#1-Fair Division with Social Impact}}

\newcommand{\Oh}[1]{{\mathcal{O}\left(#1\right)}}

\NewDocumentCommand{\cc}{ O{} O{} m }{\mbox{%
    \expandafter\ifx\expandafter\relax\detokenize{#2}\relax\else{#2-}\fi%
    \textsf{#3}%
    \expandafter\ifx\expandafter\relax\detokenize{#1}\relax\else{-#1}\fi%
    }\xspace}
\newcommand{\NP}{\cc{NP}}
\newcommand{\NPh}{\cc[hard]{NP}}
\newcommand{\NPhness}{\cc[hardness]{NP}}
\newcommand{\NPc}{\cc[complete]{NP}}

\newcommand{\FPT}{\cc{FPT}}

\newtheorem{theorem}{Theorem}
\newtheorem{lemma}{Lemma}

\newtheorem{observation}{Observation}
\newtheorem{proposition}{Proposition}
\newtheorem{definition}{Definition}
\newtheorem{example}{Example}
\newtheorem{remark}{Remark}

\begin{document}

\fancypagestyle{firstpagestyle}{
    \fancyhf{}%
    \fancyfoot[RO,LE]{}%
  }%

\fancyfoot{}

\title%
{Dividing Indivisible Items for the Benefit of All: It is Hard to Be Fair Without Social Awareness%
}
\titlenote{%
    An extended abstract of this work has been published in the Proceedings of the 40th {AAAI} Conference on Artificial Intelligence, AAAI~'26~\cite{DeligkasEGKS2026}.}%

\author{Argyrios Deligkas}
\orcid{0000-0002-6513-6748}
\email{argyrios.deligkas@rhul.ac.uk}
\affiliation{%
    \institution{Royal Holloway, University of London}
    \city{Egham}
    \country{United Kingdom}
}

\author{Eduard Eiben}
\orcid{0000-0003-2628-3435}
\email{eduard.eiben@rhul.ac.uk}
\affiliation{%
    \institution{Royal Holloway, University of London}
    \city{Egham}
    \country{United Kingdom}
}

\author{Tiger-Lily Goldsmith}
\orcid{0000-0003-0458-6267}
\email{tiger-lily.goldsmith@rhul.ac.uk}
\affiliation{%
    \institution{Royal Holloway, University of London}
    \city{Egham}
    \country{United Kingdom}
}

\author{Dušan Knop}
\orcid{0000-0003-2588-5709}
\email{dusan.knop@fit.cvut.cz}
\affiliation{
    \institution{Czech Technical University in Prague}
    \city{Prague}
    \country{Czech Republic}
}

\author{Šimon Schierreich}
\orcid{0000-0001-8901-1942}
\email{schiesim@fit.cvut.cz}
\affiliation{
    \institution{Czech Technical University in Prague}
    \city{Prague}
    \country{Czech Republic}
}
\affiliation{
    \institution{AGH University of Krakow}
    \city{Krakow}
    \country{Poland}
}

\renewcommand{\shortauthors}{Deligkas, Eiben, Goldsmith, Knop \& Schierreich}

\begin{abstract}
    In standard fair division models, we assume that all agents are selfish. However, in many scenarios, division of resources has a direct impact on the whole group or even society. Therefore, we study fair allocations of indivisible items that, at the same time, maximize {\em social impact}.
    In this model, each agent is associated with two additive functions that define their value and social impact for each item. The goal is to allocate items so that the social impact is maximized while maintaining some fairness criterion. We reveal that the complexity of the problem heavily depends on whether the agents are socially aware, i.e., they take into consideration the social impact functions. For socially unaware agents, we prove that the problem is \NPh for a variety of fairness notions, and that it is tractable only for very restricted cases, e.g., if, for every agent, the valuation equals social impact {\em and} it is binary. On the other hand, social awareness allows for fair allocations that maximize social impact, and such allocations can be computed in polynomial time. Interestingly, the problem becomes again intractable as soon as the definition of social awareness is relaxed.
\end{abstract}

\maketitle

\section{Introduction}

The new equipment for your lab has arrived, and the lab director (LD) has to decide how to allocate it among lab members. Each member has their own preferences, i.e., {\em subjective valuations}, over the items of equipment. Meanwhile, LD anticipates the {\em impact} each member can have for the {\em whole} lab when allocated an item. This can depend on the background and technical skills of each member, and thus it can vary significantly {\em and} might not be aligned with members' valuations. Having said this, LD knows that the members will {\em compare} their bundle against their peers. Thus, LD would ideally like to find an allocation that maximizes social impact {\em and} is perceived as being {\em ``fair''}.

Consider the following scenario, though. Two items have to be allocated between reliable-Bill and sloppy-Joe. It is known that the impact of Bill will be 10 from each item, while Joe can contribute only 1 for each. 
This toy example demonstrates that even the most relaxed fairness criterion -- envy-freeness up to one item (EF1) -- cannot always be satisfied if we require maximization of social impact. In many cases in the real world, though, people are {\em socially aware}. Thus, they are willing to accept a ``non-fair'' outcome if they realize that this is for the {\em benefit of all}. 

The scenarios above might be seen as \emph{fair division with externalities}~\cite{velez2016fairness,aziz2023fairness,DeligkasEKS2024}, where agents care not only for the items assigned to them, but also how exactly the remaining items are allocated. 
Nevertheless, this model is too general, so 
\citet{FlamminiGV2025}
have recently introduced a model that captures exactly the setting above: there is a set of goods and every agent is associated with two additive functions that define their value and social impact for each good. 
The goal is to allocate the goods such that the social impact is maximized while maintaining some fairness criterion.
They have studied 
the ``{\em price of fairness}'', which essentially quantifies how much worse in terms of social welfare a fair outcome is. 
In addition, they have formally introduced the notion of socially aware agents 
and which outcomes they perceive as fair.
They proved that with socially aware agents, EF1 and social impact maximizing allocations always exist and can be computed in polynomial time.
However, the complexity of the underlying computational problems was not examined in detail. 
Our goal is to extend the fairness solution concepts for this model and settle their complexity.

\subsection{Our Contribution}

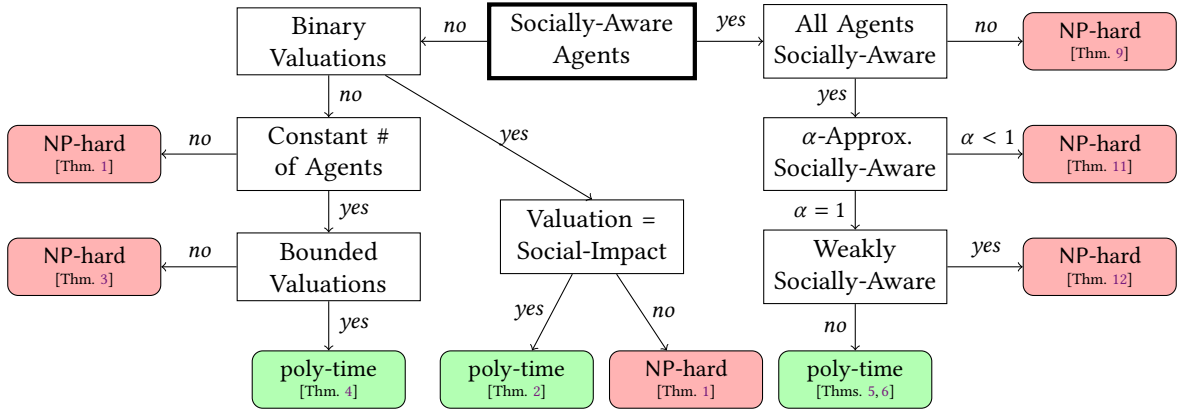
\begin{figure*}
    \centering
    \begin{tikzpicture}[align=center,node distance=1.5cm]
        \tikzstyle{decision} = [draw,align=center,text width=2.2cm];
        \tikzstyle{result} = [draw,rectangle,rounded corners,align=center,text width=1.8cm,node distance=3.25cm,font=\tiny\linespread{1.1}\selectfont];
        \tikzstyle{NPh} = [fill=red!30];
        \tikzstyle{poly} = [fill=green!30];
        \renewcommand{\NPc}{\NPh}

        \node[decision,text width=2.5cm,ultra thick] (d1) at (0,0) {\mbox{Socially-Aware} Agents};

        \node[decision] (d2) at (3.5,0) {All Agents \mbox{Socially-Aware}};
        \draw[->] (d1) -- (d2) node [above,midway] {\it\small yes};
        
        \node[result,NPh,right of=d2] (r1) {\small\NPc\\\tiny[Thm.~\ref{thm:oneSociallyUnaware:NPh}]};
        \draw[->] (d2) -- (r1) node [above,midway] {\it\small no};

        \node[decision,below of=d2,node distance=1.5cm] (d3) {\mbox{$\alpha$-Approx.} \mbox{Socially-Aware}};
        \draw[->] (d2) -- (d3) node [left, midway] {\it\small yes};
        
        \node[result,NPh,right of=d3] (r2) {\small\NPc\\\tiny[Thm.~\ref{thm:alphaSA:NPc}]};
        \draw[->] (d3) -- (r2) node [above,midway] {\it\small~$\alpha < 1$};

        \node[decision,below of=d3,node distance=1.5cm] (d4) {Weakly \mbox{Socially-Aware}};
        \draw[->] (d3) -- (d4) node [left, midway] {\it\small$\alpha=1$};

        \node[result,NPh,right of=d4] (r3) {\small\NPc\\\tiny{[Thm.~\ref{thm:weaklySA:NPh}]}};
        \draw[->] (d4) -- (r3) node [above,midway] {\it\small yes};

        \node[result,poly,below of=d4,node distance=1.5cm] (r4) {\small poly-time\\\tiny[Thms.~\ref{thm:sa_swef1_goods},\,\ref{thm:EFL_poly}]};
        \draw[->] (d4) -- (r4) node [left, midway] {\it\small no};

        \node[decision] (d5) at (-3.5,0) {Binary Valuations};
        \draw[->] (d1) -- (d5) node [above,midway] {\it\small no};

        \node[decision,below of=d1,node distance=2.6cm] (d6) {Valuation = Social-Impact};
        \draw[->] (d5) -- (d6.north) node [right,midway] {\it\small yes};

        \node[result,poly,below of=d6,node distance=1.9cm,xshift=-1cm] (r5) {\small poly-time\\\tiny[Thm.~\ref{thm:binary_equal_polytime}]};
        \draw[->] (d6) -- (r5) node [left, midway] {\it\small yes};
        \node[result,NPh,right of=r5,node distance=2.25cm] (r6) {\small\NPc\\\tiny{[Thm.~\ref{thm:NPh_binary}]}};
        \draw[->] (d6) -- (r6) node [right, midway] {\it\small no};

        \node[decision,below of=d5,node distance=1.5cm] (d7) {Constant \# of Agents};
        \draw[->] (d5) -- (d7) node [right,midway] {\it\small no};

        \node[result,NPh,left of=d7] (r7) {\small\NPc\\\tiny{[Thm.~\ref{thm:NPh_binary}]}};
        \draw[->] (d7) -- (r7) node [above,midway] {\it\small no};

        \node[decision,below of=d7,node distance=1.5cm] (d8) {Bounded Valuations};
        \draw[->] (d7) -- (d8) node [right,midway] {\it\small yes};

        \node[result,NPh,left of=d8] (r8) {\small\NPc\\\tiny{[Thm.~\ref{thm:two_agents_hard}]}};
        \draw[->] (d8) -- (r8) node [above,midway] {\it\small no};
        \node[result,poly,below of=d8,node distance=1.5cm] (r9) {\small poly-time\\\tiny{[Thm.~\ref{UMandEF1:pseudopolynomialConstantNumberOfAgents}]}};
        \draw[->] (d8) -- (r9) node [right,midway] {\it\small yes};
    \end{tikzpicture}
    \caption{A basic overview of the complexity landscape of deciding whether a SIM and~$\mathcal{F}$-fair allocation exist.}
    \Description{A decision diagram providing a map to our complexity results.}
    \label{fig:resultsOverview}
\end{figure*}

We provide a comprehensive study of the complexity of computing fair allocations that maximize social impact, which we term SIM allocations. We investigate seven different fairness criteria, all of which are strengthenings or weakenings of EF1. For each criterion, we examine the constraints on the valuation and social impact functions that allow for tractability (see \Cref{fig:resultsOverview} for an overview of our results).

Our initial set of technical results considers socially unaware agents, where we study two different dimensions of the problem and for each of them we provide complementary results between tractability and \NPhness. Interestingly, we show that the problem behaves in the same way for every fairness criterion, albeit we need slight modifications in our proofs to formally establish this. First, we restrict the domain of the valuation and social impact functions. There, we show that the problem is \NPh even when both functions are binary  (\Cref{thm:NPh_binary}), but it becomes tractable if for every agent both functions are binary {\em and} equal to each other (\Cref{thm:binary_equal_polytime}).

Then, we restrict the number of agents. Unfortunately, the problem is \NPh even for two agents with binary social impact functions (\Cref{thm:two_agents_hard}). On the other hand, we derive a pseudo-polynomial time algorithm that solves the problem for any constant number of agents (\Cref{UMandEF1:pseudopolynomialConstantNumberOfAgents}). Hence, the problem is tractable when we have a constant number of agents, polynomially bounded valuation functions, and unrestricted social impact functions.

Then we consider the case where the agents are socially aware. Our first set of results for this setting adopts the definition of social awareness from \cite{FlamminiGV2025}. Intuitively, we accept allocations that do not satisfy the ``standard'' fairness constraint, however, any bundle an agent envies (according to the chosen fairness criterion) cannot produce {\em strictly} more social impact if this agent gets it. We prove that social awareness allows for SIM and fair allocations that can be computed in polynomial time, for {\em all} fairness criteria we study (\Cref{thm:sa_swef1_goods,thm:EFL_poly}). This leads us to the question of whether the definition of social awareness allows for an arbitrary fairness criterion. We prove that this is not the case by showing \NPhness if we have ``{\em fully envious}'' agents (\Cref{thm:czechEnvy:NPh}).

Finally, we ask whether we can relax social awareness and still get a fair and SIM allocation in polynomial time. Our first relaxation considers the case where not all agents are socially aware. Unfortunately, even if only one agent is socially unaware existence of SIM and fair allocations is not guaranteed and, in fact, the problem becomes \NPh even when there are two agents (\Cref{thm:oneSociallyUnaware:NPh}). On the positive side, we complement this negative result with a pseudo-polynomial time algorithm that works for any constant number of agents where some of them are socially unaware (\Cref{thm:socially_Unaware_pseudopoly}).
Then, we limit social awareness for each agent in two different ways.
In the first case, we focus on {\em~$\alpha$-approximately} socially aware agents, where~$\alpha \in [0,1]$, and an agent accepts an unfair outcome only if every envied bundle produces at least~$\alpha$ times more social impact under the current allocation. 
Strikingly, we get the same behavior as before: there is not always a solution and the problem is \NPh for {\em every} possible~$\alpha$ (\Cref{thm:alphaSA:NPc}).
In the second case, we define {\em weakly socially aware} agents. Intuitively, an agent accepts an unfair allocation if the proportional gain of their valuation from an envied bundle is less than the proportional gain in social impact (\Cref{def:wsa}). Once again, there is no guaranteed solution and the problem is \NPh even for two agents (\Cref{thm:weaklySA:NPh}).

\subsection{Related Work}Our paper lies in the area of fair division with indivisible items, where there is a flourish of results in the last 15 years~\citep{LiptonMMS2004,BouveretL08,Budish2011,CaragiannisKMPSW19}; for excellent recent surveys, see~\citep{AmanatidisABFLMVW2023} and~\citep{NguyenR2023}.

As we have mentioned before, \cite{FlamminiGV2025} is the most relevant paper to ours, since we extend their model and their computational results. Having said this, both papers fall into the larger, more general model of {\em fair division with externalities}~\cite{velez2016fairness,aziz2023fairness,DeligkasEKS2024}. In this model, each agent gets value from the {\em whole} allocation, i.e., they value each bundle depending on the agent who got it. Hence, we can see any instance in our setting as an instance of agents with externalities. For socially unaware agents, we create an instance with externalities with one extra agent that has zero value for every item they get, while for socially aware agents, each agent has externalities as well that match the social impact.
On a similar note, \citet{BuLLST2023} studied the model in which there is an \emph{allocator} with its own view of how the items should be allocated between agents, and the goal is to be fair to both the agents and the allocator.
It should be highlighted that users deriving their utility not only from their own benefit but also from the benefit of (or to) others is a well-described phenomenon in social sciences and behavioral economics; see, e.g., \cite{Overvold1980,Andreoni1990,HuangLET2022,ThomasT2025}.

A different point of view is to consider our problem as {\em constrained} fair division~\cite{suksompong2021constraints}. We identify two settings that are closest to ours. The first regards {\em welfare-maximizing} fair allocations. Observe that when social impact equals the valuation of each agent, fair SIM allocations coincide with welfare-maximizing allocations that are fair, which was shown to be \NPh~\cite{AzizHMS2023} in general. The second is related to {\em orientations}~\cite{ChristodoulouFKS_orientations,ZhouWLL2024,DEGK2024_orientations,AfshinmehrDKMR2025,BlazejGSS2025}. Here, every agent is associated with a subset of items that they are interested in. The goal is to find an allocation where each agent gets items {\em only} from their subset. In our case, agents are only allowed to get items for which they are social impact maximizers. The crucial difference though is that in our case, agents value the remaining items as well.

\section{Preliminaries}\label{sec:prelims}

We assume a set~$\items$ of~$\numItems$ \emph{items} which need to be partitioned between a set~$\agents$ of~$\numAgents$ \emph{agents}. Each agent~$i\in\agents$ is associated with a \emph{valuation function}~$\val_i\colon 2^\items\to\N_0$ which assigns to each subset~$S$ of items, called a \emph{bundle}, a numerical value expressing how satisfied~$i$ is when~$S$ is allocated to them. Additionally, there is a \emph{social impact} function~${\si_i\colon 2^\items\to\N_0}$ for every agent~$i\in\agents$ expressing what social impact agent~$i$ generates for the society with each bundle~$S\subseteq \items$. Throughout the paper, we assume that both the valuation functions and the social impact functions are \emph{additive}, meaning that for every agent~$i\in\agents$ and every~$S\subseteq\items$ we have~$\val_i(S) = \sum_{g\in S} \val_i(\{g\})$ and~$\si_i(S) = \sum_{g\in S} \si_i(\{g\})$, respectively, and \emph{normalized}, i.e.,~$\val_i(\emptyset) = \si_i(\emptyset) = 0$. 
If an item~$g\in\items$ is valued (in terms of~$v_i$'s) non-negatively by every agent, we say that~$g$ is \emph{good}. An item is called a \emph{chore} if it is valued non-positively by every agent. \textbf{Unless stated otherwise, we assume that all items are goods.} 
An instance of \emph{fair division with social impact} is then a quadruple~$\mathcal{I} = (\agents,\items,(\val_i)_{i\in\agents},(\si_i)_{i\in\agents})$.

A \emph{partial allocation} is a tuple~$\alloc = (\bundle_1,\ldots,\bundle_\numAgents)$, where~$\bundle_i\subseteq M$ is a (potentially empty) bundle received by agent~$i\in\agents$, for which it holds that~$\bundle_i \cap \bundle_j = \emptyset$ for each pair of distinct~$i,j\in\agents$. If additionally~$\bigcup_{i=1}^\numAgents \bundle_i = \items$, then we say~$\alloc$ is a \emph{complete allocation}. In the remainder of the paper, we are interested in complete allocations. Let~$\alloc$ be a (partial) allocation. 
We naturally extend the valuations and social impact functions from bundles to allocations by setting~$\val_i(\alloc) = \val_i(\bundle_i)$ and~$\si_i(\alloc) = \si_i(\bundle_i)$.
Moreover, we denote the overall social impact generated by~$\alloc$ as \[\SI(\alloc) = \sum_{i\in\agents} s_i(\bundle_i).\]

\subsection{Solution Concepts} 

The first solution concept crucial for our work is the requirement to maximize social impact. Specifically, we require that in each solution the sum of social impacts is the maximum possible.

\begin{definition}
    An allocation~$\alloc$ is \emph{social impact maximizing} (SIM) if for every allocation~$\alloc'$ we have~$\SI(\alloc) \geq \SI(\alloc')$.
\end{definition}

It is easy to see that SIM allocations always exist and can be found efficiently. It is enough to allocate each item to an agent that maximizes the social impact for it. However, such allocations may be very unfair. Assume an instance with social impact functions that assign one to every combination of an agent and an item. It may happen that the previous procedure allocates all items to a single agent~$i$, which is arguably very unfair to all remaining agents. 

To overcome this issue, we combine SIM requirement with various notions of fairness. We focus on axioms based on pairwise comparison of agents' bundles. The prototypical concept in this line of research is \emph{envy-freeness}~\cite{Foley1967}, which requires that each agent prefers their own bundle before the bundles of other agents and is defined as follows.

\begin{definition}
    An allocation~$\alloc = (\bundle_1,\ldots,\bundle_\numAgents)$ is called \emph{envy-free} (EF) if for any pair of agents~$i,j\in\agents$ we have~$\val_i(\bundle_i) \geq \val_i(\bundle_j)$.
\end{definition}

It is well known that envy-free allocations are not guaranteed to exist, even without social impact functions. Simply assume an instance with two agents and one item strictly positively valued by both agents; the agent who does not receive this item is always envious. Moreover, deciding the existence of EF allocations is a notoriously hard computational problem. Therefore, in the rest of the paper, we focus on several relaxations of envy-freeness for which, without social impact function, existence is guaranteed and there are efficient algorithms finding such desirable outcomes.

We start with a relaxation called \emph{envy-freeness up to one item}~\cite{LiptonMMS2004,Budish2011} which requires that if there is envy from agent~$i$ to agent~$j$, then this envy can be eliminated by the removal of one item from~$j$'s bundle.

\begin{definition}
    An allocation~$\alloc = (\bundle_1,\ldots,\bundle_\numAgents)$ is called \emph{envy-free up to one item} (EF1) if for every pair of agents~$i,j\in\agents$ there exists an item~$g\in \items$ such that~$\val_i(\bundle_i) \geq \val_i(\bundle_j\setminus\{g\})$.
\end{definition}

Note that in the definition of EF1, the item removed from~$j$'s bundle can be \emph{specific} for each agent~$i\in\agents$. A stronger variant of this notion requires that the removal of one \emph{universal} item from~$\bundle_j$ prevents potential envy from all other agents. This notion is due to \citet{ConitzerFSW2019}.

\begin{definition}
    An allocation~$\alloc = (\bundle_1,\ldots,\bundle_\numAgents)$ is called \emph{strongly envy-free up to one item} (sEF1) if for every~$j\in\agents$ there exists an item~$g\in \items$ such that~$\val_i(\bundle_i) \geq \val_i(\bundle_j\setminus\{g\})$ for every~$i\in\agents$.
\end{definition}

As noted by \citet{WuZZ2025}, in traditional EF1 (or sEF1), it is assumed that all agents have equal obligations. However, this assumption can be very unrealistic in many scenarios. E.g., recall our scientific laboratory example from the beginning of the paper. It is reasonable to expect that the division of the equipment should also consider factors such as seniority or merit. Therefore, we also study \emph{weighted envy-freeness}~\cite{ChakrabortyISZ2021,AzizLMWZ2024,Suksompong2025}. In this model, each agent~$i\in\agents$ is additionally associated with its \emph{weight}~$\wFn_i$. The concept is then defined as follows.

\begin{definition}
    An allocation~$\alloc = (\bundle_1,\ldots,\bundle_\numAgents)$ is called \emph{weighted envy-free up to one item} (wEF1) if for every pair of agents~$i,j\in\agents$ there exists an item~$g\in\items$ such that
    \[
        \frac{\val_i(\bundle_i)}{\wFn_i} \geq \frac{\val_i(\bundle_j\setminus\{g\})}{\wFn_j}\,.
    \]
\end{definition}

Clearly, if all the weights are the same, wEF1 is equivalent to EF1. Therefore, wEF1 is a generalization of EF1. Similarly to sEF1, we can analogously define \emph{strong weighted envy-freeness up to one item} (swEF1).

\begin{definition}
    An allocation~$\alloc = (\bundle_1,\ldots,\bundle_\numAgents)$ is called \emph{strongly weighted envy-free up to one item} (swEF1) if for every~$j\in\agents$ there exists an item~$g\in\items$ such that for every~$i\in\agents$ we have
    \[
        \frac{\val_i(\bundle_i)}{\wFn_i} \geq \frac{\val(\bundle_j\setminus\{g\})}{\wFn_j}\,.
    \]
\end{definition}

A different point of criticism towards EF1 is that it allows the envy between~$i$ and~$j$ to vanish simply by the removal of the best good from the~$j$'s bundle, regardless of how valued this item is for the agent~$i$. For example, assume that~$i$ and~$j$ are allocated one item each, an~$i$ values~$\bundle_j$ as~$1000$ and~$\bundle_i$ as~$1$. Although formally there is no EF1-envy from~$i$ towards~$j$, agent~$i$ may still consider the allocation very unfair, since there is a huge difference in the valuations of both bundles. Motivated by this, \citet{BarmanBMN2018} introduced \emph{envy-freeness up to one less preferred good}.

\begin{definition}
    An allocation~$\alloc = (\bundle_1,\ldots,\bundle_\numAgents)$ is called \emph{envy-free up to one less preferred item} (EFL) if for every pair of agents~$i,j\in\agents$ at least one of the following two conditions hold:
    \begin{enumerate}
        \item~$A_j$ contains at most one item which is positively valued by~$i$, or
        \item there exists~$g\in A_j$ such that~$\val_{i}(A_i) \geq \val_i(A_j\setminus \{g\})$ and~$\val_i(A_i) \geq \val_i(\{g\})$.
    \end{enumerate}
\end{definition}

It follows directly from the definition that every EFL allocation is simultaneously EF1.

Finally, in light of our negative results, we also study a relaxation of EF1 that, instead of removing one item to eliminate envy, uses a transfer of one item. In the literature, this notion is usually called weak envy-freeness up to one item. However, to avoid possible confusion with weighted EF1, we call it \emph{envy-freeness up to one transfer}.

\begin{definition}
    An allocation~$\alloc = (\bundle_1,\ldots,\bundle_\numAgents)$ is called \emph{envy-free up to one transfer} (tEF1) if for every pair of agents~$i,j\in\agents$ such that~$\val_i(\bundle_i) < \val_i(\bundle_j)$ there exists an item~$g\in \bundle_j$ so that~$\val_i(\bundle_i\cup\{g\}) \geq \val_j(\bundle_j\setminus\{g\})$.
\end{definition}

It is easy to see that whenever an allocation is EF1, it is tEF1, as if the item~$g$ that removes the envy from~$i$ to~$j$ is additionally moved to~$\bundle_i$, it cannot decrease the value of~$\bundle_i$ for~$i$. In the opposite direction, tEF1 does not imply EF1. The example here is an allocation where~$\bundle_j$ contains two items, both valued~$1$ by~$i$, and~$\bundle_i = \emptyset$. If we transfer one item from~$\bundle_j$ to~$\bundle_i$, the value of agent~$i$ for both bundles is~$1$, so this allocation is tEF1. However, the removal of a single item does not eliminate envy, so it is not EF1.

\begin{remark}
    It is not hard to see that given an allocation~$\alloc$, one can check whether this allocation is SIM and~\mbox{$\mathcal{F}$-fair}, where~$\mathcal{F}\in\{\text{tEF1},\text{EF1},\text{sEF1},\text{wEF1},\text{swEF1},\text{EFL}\}$, in polynomial time. Hence, deciding whether a SIM and~$\mathcal{F}$-fair allocation exist is trivially in \NP and we will not stress this explicitly in our hardness proofs.
\end{remark}

\section{Socially Unaware Agents}\label{sec:SA}

In this section, we explore the computational complexity of deciding the existence of SIM and fair allocations. In general, hardness for EF1 follows from the result of \citet{AzizHMS2023}, who showed that deciding whether an EF1 allocation maximizing social-welfare\footnote{Social welfare of an allocation is simply the sum of utilities of all agents in this allocation.} exists is \NPc. This corresponds to the case with~$\si_i = \val_i$ for every~$i\in\agents$. 
We extend their result by studying also additional notions of fairness and giving a much more detailed picture of its complexity, even in as restricted domains as binary valuations, or with a constant number of agents.

First, we observe that, without loss of generality, we can assume that the social impact function is binary.

\begin{lemma}
    \label{UMandEF1:socialImpactCanBeBinary}
    For any~$\mathcal{F}\in\{\text{tEF1},\text{EF1},\text{sEF1},\text{wEF1},\text{swEF1},\text{EFL},\text{EF}\}$, an instance~$\mathcal{I} = (\agents,\items,(\val_i)_{i\in\agents},(s_i)_{i\in\agents})$ of fair division with social impact admits a SIM and~$\mathcal{F}$-fair allocation if and only if an instance~$\mathcal{J}=(\agents,\items,(\val_i)_{i\in\agents},(\si^*_i)_{i\in\agents})$, where~${\si^*_i(g) = 1}$ if~${i\in \arg\max_{i\in\agents} \si_i(g)}$ and~$0$ otherwise, admits a SIM and~$\mathcal{F}$-fair allocation.
\end{lemma}
\begin{proof}
    For our problem, we have to maximize the sum of social impact in a solution; regardless of the fairness notion/solution concept. Hence, we can consider any~$\mathcal{F}$ of \FDSI for this (since the part of the instance that is changing is the social impact ``scores"). 
    
    Consider that we have a solution for~$\mathcal{I}$. In a SIM solution, agents who maximize the social impact for a single item must receive this in the allocation because the social impact function is additive. In the same allocation in the corresponding instance~$\mathcal{J}$, the solution must also be SIM because some agent~$i$ can only receive some item~$g$ if they are a social impact maximizer, i.e.,~$s^*_i(g) = 1$. Which, by definition, means they are the social impact maximizer in~$\mathcal{I}$. 
    
    On the other hand, consider we have a solution for instance~$\mathcal{J}$ where the social impact function is binary. Consider now that finding a SIM solution in the binary instance did not give us a SIM assignment in the original instance. Therefore, this means there exists some agent~$i$ that has a higher social impact for an item~$g$ than agent~$j$ who received it. 
    If this was the case, agent i would not have been assigned~$v_i(s)= 1$ for~$g$ in the binary instance and hence could not have received the item. Hence, contradiction. Therefore it must be a solution in~$\mathcal{I}$.

    Note that, in both directions, the valuation function~$v_i$ of agent~$i$ remains the same. Therefore, if we have a solution in~$\mathcal{I}$ which is~$\mathcal{F}$-fair the same allocation is also~$\mathcal{F}$-fair in~$\mathcal{J}$, vice versa, because the fairness notion only depends on the agents' valuations.
\end{proof}

With the previous lemma in hand, we show our main hardness result of this section, which settles that even if both the valuations and the social impact functions are binary, the problem is \NPc. Our strategy in the reduction is to introduce some initial unfairness between agents so that the rest of the items have to be allocated in an envy-free way, which is known to be computationally hard~\cite{AzizGMW2015,HosseiniSVWX2020}.

\begin{theorem}\label{thm:NPh_binary}
    For any~$\mathcal{F}\in\{\text{tEF1},\text{EF1},\text{sEF1},\text{wEF1},\text{swEF1},\text{EFL}\}$, it is \NPc to decide whether a SIM and \mbox{$\mathcal{F}$-fair} allocation exist, even if both the valuations and the social impact functions are binary.
\end{theorem}
\begin{proof}
    \citet[Theorem 11]{AzizGMW2015}, show that it is \NPh to find an envy-free allocation of items even when agents have binary additive valuations.
    
    In an instance of envy-free allocation~$\mathcal{I}$ we have a set of items~$M$ and a set of agents~$N$. The goal is to find an envy-free allocation.
    
    From~$\mathcal{I}$ we create an instance~$\mathcal{J}$ of \FDSI. We create some additional items to add to the set~$M$. We will call the set of these special items~$M'$.
    For each agent~$i \in N$ we create an item~$g_i$ with~$v_i(g_i) = 0$ and~$s_i(g_i)=1$. For all other agents~$i \neq j$, let~$v_j(g_i) = 1$ and~$s_j(g_i)=0$. %
    
    As for the remaining “standard" items, i.e., the set~$M$, the valuation functions remain the same. For each agent~$i \in N$ let~$s_i(g)=1$ for all~$g \in M$. 
    
    Observe that in the solution, every agent~$j$ must receive his special items~$g_j$ because he is the only agent who maximizes social impact for these. 
    Note that in a SIM solution, any agents can receive the items in~$M$. as they all have social impact 1. 
    
    \noindent\textbf{EF1}\hspace{0.25cm}
    Given a solution to~$\mathcal{I}$, we have an allocation of the items in~$M$ which is envy-free; what remains to show is that by assigning items the special items in~$M'$ we still end up with an EF1 allocation. 
    Since we must have a solution which is SIM and EF1, all special items must be given to their respective agents for it to be SIM, i.e., we give item~$g_i$ to the agent~$i$ who is the only agent with~$s_i(g_i)=1$. 
    This is clearly an EF1 allocation, as no agent was envious prior to this assignment and every agent has received exactly one item. (Therefore, if any agent became envious, it must be envy-free after removing one good.)
    
    On the other hand, take a solution to~$\mathcal{J}$. Each special item must be allocated to its respective agents, since we are given a SIM solution. Without loss of generality, we can assume that~$M$ is non-empty. For each agent~$i$, he must get allocated at least one standard item in addition to~$g_i$ or else he will EF1-envy (i.e., EF1 is not satisfied for him) as his bundle will be of value~$0$ for him.
    Thus, this allocation of items in~$M$ must be envy-free, because already due to the special items agent~$i$ has to remove~$g_j$ in order not to envy~$j$, i.e., they cannot afford to envy agent~$j$ for any more items.
    
    \noindent\textbf{EFL, sEF1,wEF1,swEF1,}\hspace{0.25cm}
    Since we have binary valuations, EF1 and EFL coincide. 
    Strong EF1 requires all agents who envy agent~$j$ to be removing the same item. The above construction straightforwardly holds for sEF1, as~$g_i$ can always be the item that agents are removing.
    In addition, EF1 is a special case of wEF1 (likewise, sEF1 for swEF1); therefore, the construction follows.
    
    \noindent\textbf{tEF1}\hspace{0.25cm}
    We make a small modification to the construction. For each agent, we create two copies of the special item, referred to~$g_i^1$ and~$g_i^2$. The key difference is that there are now {\emph {two}} items that must go to each agent, and every other agent will envy them. 
    
    Let~$v_i^o(A_i)$ be agent~$i$'s valuation function for instance~$\mathcal{I}$ where he receives the bundle~$A_i$. 
    Let~$v_i^r(A_i')$ be agent~$i$'s valuation function in instance~$\mathcal{J}$ where he receives the bundle~$A_i'$. 
    
    In our construction, we have:~$v_i^o(A_i) = v_i^r(A_i')$ and~$v_i^o(A_j) = v_i^r(A_j') -2$ as in~$\mathcal{I}$ they have received~$g_i^1$ and~$g_i^2$.
    
    Moreover, given a solution to~$\mathcal{I}$, we have that~$v_i^o(A_i) \geq v_i^o(A_j)$. Therefore by giving out the special items to their respective agents, we get a SIM solution which is tEF1 because the following holds:~$v_i^r(A_i'\cup\{g_i^1\} = v_i^o(A_i)+1 \geq v_i^o(A_j)+1 = v_i^r(A_j' \setminus\{g_j^1\}) = v_i^r(A_j')-1$. 
    
    Given a solution to~$\mathcal{J}$, we have a SIM solution where~${g_i^1, g_i^2} \in A_i'$. Hence, for a pair of agents~$i$ and~$j$ the following holds because the allocation is tEF1:
   ~$v_i^o(A_i) = v_i^r(A_i' \cup \{g_j^1\}-1) \geq v_i^r(A_j'\setminus\{g_j^1\})-1 = v_i^o(A_j)+1-1 = v_i^o(A_j)
   ~$, i.e., agent~$i$'s valuation for~$A_i'$ is greater than his valuation for~$A_j$ if he transfers~$g_j^2$, therefore~$v_i^o(A_i) \geq v_i^o(A_j)$ (envy-free).
\end{proof}

Note that in the previous results, the social impact functions and the valuations are not the same. If we additionally assume that the valuations and social impacts are the same, we obtain a positive result modifying the Yankee-swap algorithm of \citet{viswanathan2023yankee} and the weighted round-robin of \citet{ChakrabortyISZ2021}, respectively.

\begin{theorem}\label{thm:binary_equal_polytime}
    If the valuations are binary and~$\val_i = s_i$ for every~$i\in\agents$, a SIM and~$\mathcal{F}$-fair allocation is guaranteed to exist and can be found in polynomial time for any~$\mathcal{F}\in\{\text{tEF1},\text{EF1},\text{sEF1},\text{wEF1},\text{swEF1},\text{EFL}\}$.
\end{theorem}
\begin{proof}
    \noindent\textbf{EF1, EFL, sEF1, tEF1.}\hspace{0.25cm}
    In their paper, \citet{viswanathan2023yankee} prove that an EFX\footnote{For this result, we stress here that \citet{viswanathan2023yankee} assume EFX definition where envy can be eliminated by the removal of \emph{any} item (even those valued~$0$ by the agent) as defined by \citet{PlautR2020}. This notion is sometimes also referred to as EFX\textsubscript{0}.} allocation with maximum utilitarian social welfare can be found in polynomial time using the algorithm they introduce: Yankee Swap. Observe that it is SIM because~$v_i = s_i$, i.e., a maximum allocation with respect to the agents' valuation functions is also a SIM solution.
    Thus, for our problem, because EFX~$\implies$ EFL~$\implies$ EF1~$\implies$ tEF1 and EFX~$\implies$ sEF1, we have existence and a polynomial time algorithm.
    
    \noindent\textbf{wEF1, swEF1.}\hspace{0.25cm}
    In their paper, \citet{ChakrabortyISZ2021} introduce an algorithm to find swEF1 allocations in polynomial time. In this algorithm, agents select items in a weight-based picking sequence, and after each round, the partial allocation is swEF1.
    No agent ever receives an item of value~$0$, as they can stop picking once there is no item that they value. Hence, since~$v_i=s_i$ the allocation is SIM. 
    Since swEF1~$\implies$ wEF1, we have existence and a polynomial time algorithm for both notions.
\end{proof}

Next, we turn our attention to instances with a constant number of agents. We show that with two agents already, deciding the existence of SIM and fair allocation is computationally intractable.

\begin{theorem}\label{thm:two_agents_hard}
    For any~$\mathcal{F}\in\{\text{tEF1},\text{EF1},\text{sEF1},\text{wEF1},\text{swEF1},\text{EFL}\}$, it is \NPc to decide whether a SIM and~$\mathcal F$-fair allocation exist, even if the social impact functions are binary and~$\numAgents = 2$.
\end{theorem}
\begin{proof}
    Our reduction is from the \probName{Partition} problem. There, we are given a multi-set of integers~$W=\{w_1,\ldots,w_\ell\}$ such that~$\sum_{j\in[\ell]} w_j = 2t$ and the goal is to decide whether a set~$J\subseteq [\ell]$ exists so that~$\sum_{j\in J} w_j = \sum_{j\in [\ell]\setminus J} w_j = t$. The problem is known to be \NPh, even if no~$w_j$,~$j\in[\ell]$, is greater than or equal to~$t$~\cite{GareyJ1979}.

    Given an instance~$\mathcal{J}$ of \probName{Partition}, we construct an equivalent instance~$\mathcal{I}$ of of fair division with social impact as follows (for the sake of exposition, we start with the case of EF1 allocations; we show how the reduction works for other fairness notions assumed at the end of the proof). There are two agents,~$a_1$ and~$a_2$. The set of items~$\items$ contains one item~$g_j$ for every integer~$w_j\in W$ and two additional large items~$G_1$ and~$G_2$. The valuations and social impacts are as follows. For any~$a_i$,~$i\in[2]$, we have~$\val_{a_i}(g_j) = w_j$ and~$\si_{a_i}(g_j) = 1$. Additionally, we have~$\val_{a_1}(G_1) = \val_{a_2}(G_2) = 0$,~$\si_{a_1}(G_1) = \si_{a_2}(G_2) = 1$, and~$\val_{a_1}(G_2) = \val_{a_2}(G_1) = t$ and~$\si_{a_1}(G_2) = \si_{a_2}(G_1) = 0$. The valuation functions and social impact functions are also summarized in the following table.
    \begin{center}
        \renewcommand{\arraystretch}{1.2}
        \begin{tabular}{c|c|c|c}
                  &~$G_1$           &~$G_2$          &~$g_j$,~$j\in[\ell]$  \\\hline
           ~$a_1$ &~$\si=1, \val=0$ &~$s=0, \val=t$ &~$s=1, \val=w_j$\\
           ~$a_2$ &~$\si=0, \val=t$ &~$s=1, \val=0$ &~$s=1, \val=w_j$\\
        \end{tabular}
    \end{center}
    The idea again is that the large items introduce some initial envy, so the rest of the items have to be split equitably (or, in an envy-free way, to use the terms of our problems) between the agents.

    For correctness, let~$J$ be a solution for~$\mathcal{J}$. We construct an allocation~$\alloc=(\bundle_1,\bundle_2)$ such that~$\bundle_1 = \{G_1\} \cup \{ g_j \mid j\in J \}$ and~$\bundle_2 = \{ G_2 \} \cup \{ g_j \mid j\in[\ell]\setminus J \}$, and claim that~$\alloc$ is EF1 allocation. For agent~$a_1$, we have~$\val_{a_1}(\bundle_1) = t$ and~$\val_{a_1}(\bundle_2) = 2t$; however, the envy can be eliminated by the removal of~$G_2$. The argument for~$a_2$ is symmetric. That is,~$\alloc$ is EF1 and is obviously SIM.

    In the opposite direction, let~$\alloc=(\bundle_1,\bundle_2)$ be an EF1 and SIM allocation. Since~$\alloc$ is SIM, it must be the case that~$G_1\in \bundle_1$ and~$G_2\in \bundle_2$. Now, we construct a solution for~$\mathcal{J}$ by setting~$J = \{ j\in[\ell] \mid g_j \in \bundle_1 \}$. It remains to show that~$J$ is indeed a solution. For the sake of contradiction, let~$\sum_{j\in J} w_j < \sum_{j\in[\ell]\setminus J} w_j$ and~$\delta$ be the difference between these two sums. Then, in~$\mathcal{I}$, it holds that~$\val_{a_1}(A_2) = t + t + \delta = 2t + \delta$, while~$\val_{a_1}(A_1) = t - \delta$. Even if we remove the best item from~$\bundle_2$, the envy from~$a_1$ towards~$a_2$ persists. This contradicts that~$\alloc$ is an EF1 allocation. Hence,~$\sum_{j\in J} w_j \geq \sum_{j\in[\ell]\setminus J} w_j$. If~$\sum_{j\in J} w_j > \sum_{j\in[\ell]\setminus J} w_j$, then by symmetric arguments we obtain that~$a_2$ has EF1-envy towards~$a_1$. Thus, it must be the case that~$\sum_{j\in J} w_j = \sum_{j\in[\ell]\setminus J} w_j$, which shows that~$J$ is indeed a solution for~$\mathcal{J}$.

    \noindent{\bfseries sEF1, wEF1, swEF1.}\hspace{0.25cm}In case of two agents, sEF1 and EF1 coincides. Moreover, EF1 is a special case of wEF1 where all weights are the same, and the same is true for sEF1 and swEF1. Therefore, the construction works for all these notions without any modification.

    \noindent\textbf{tEF1.}\hspace{0.25cm} The only modification for the case of tEF1, we duplicate~$G_1$ and~$G_2$. That is, there are two copies of special items specific for each agent.

    \noindent\textbf{EFL.}\hspace{0.25cm} 
    For EFL, the reduction is identical. Observe that each agent values as~$t$ the large item he wants and that is allocated to the other agent due to social impact. Hence he has to obtain at least value of~$t$ from standard items.
\end{proof}

However, the previous hardness result is highly dependent on the fact that the agents' valuations for the items are exponential in the input size. That is, the result shows only weak \NPhness. In the following, we complement the hardness and show that if the number of agents is a fixed constant and the valuations are polynomially bounded in the input size, then there is an efficient algorithm.
Notably, the algorithm is not based on an explicit DP formulation, as one might expect, but rather models our problem as graph reachability. Thanks to this, the algorithm can be implemented very efficiently by generating the graph on-demand and finding the solution using, e.g., a variant of the A$^*$ algorithm~\cite{HartNR1968,DechterP1985,Korf1997,ZhouZ2015}.

\begin{theorem}\label{UMandEF1:pseudopolynomialConstantNumberOfAgents}
    For every fixed constant number of agents~$\numAgents$, there exists a pseudo-polynomial-time algorithm deciding whether a SIM and~$\mathcal F$-fair allocation exist for any~$\mathcal{F}\in\{\text{tEF1},\text{EF1},\text{sEF1},\text{wEF1},\text{swEF1},\text{EFL},\text{EF}\}$.
\end{theorem}
\begin{proof}
    Let~${\nu = \max_{i\in{\agents}}\max_{g\in\items} |v_i(g)|}$ be the largest value an agent can receive from a single item.
    We build a directed acyclic graph~$G = (V,E)$ with~$\numItems+2$ layers of vertices~$V_s = \{s\}, V_t = \{t\}, \{ V^i \mid i \in \items \}$.
    The vertex~$s$ is the source and~$t$ is the sink of the graph.
    Vertices in the layer~$V^i$ are in one-to-one correspondence with the set
    \[
        \{ 0, \ldots, \nu \cdot \numItems \}^{\agents \times \agents}
        \times
        \{ 0, \ldots, \nu \}^{\agents \times \agents}
        \,,
    \]
    that is, the product of 
    \begin{enumerate}
        \item a set of~$\numAgents^2$-sized tuples of numbers in~$[\nu \cdot \numItems]_0$ (this way, we store the value of the bundles as seen by all the agents), and
        \item a set of~$\numAgents^2$-sized tuples of numbers in~$[\nu]_0$ (we store the value of a largest item as seen by all the agents).
    \end{enumerate}
    The intended meaning of such a vertex~$v^i_{\bm{x}, \bm{y}}$ is that we are having a partial allocation assigning items in~$\items$ preceding~$i$ in some fixed order in such a way that the value~$x_{a,b} = \val_a(\bundle_b)$ (note that we allow~$a=b$) and the value of a best item in the bundle~$\bundle_b$ of~$b \in \agents$ as seen by~$a \in \agents$ is~$y_{a,b}$.
    We interpret the source vertex~$s$ as the vertex~$v^0_{\bm{0}, \bm{0}}$, i.e., no items assigned and all the values set to~$0$.

    We begin the exposition for the objective EF1.
    We construct the arcs of the graph in a layer by layer fashion.
    The arcs are always between consecutive layers.
    Let~$i \in \items$ be an item, we construct the arcs from the previous layer~$V^j$ to the layer~$V^i$ as follows.
    For each vertex~$v^j_{\bm{x}, \bm{y}}$ we add at most~$\numAgents$ arcs corresponding to assigning the item~$i$ to each agent.
    That is, for each agent~$c \in \agents$ we add the arc~$(v^j_{\bm{x}}, v^j_{\bm{\widehat{x}}, \bm{\widehat{y}}})$ if~$c$ is an agent maximizing the social impact if~$i$ is assigned to them,
    \[
        \widehat{x}_{a,b} =
        \begin{cases}
            x_{a,b} & a,b \in \agents \setminus \{c\} \\
            x_{c,a} & a \in \agents \setminus \{c\} \\
            x_{a,c} + \val_a(i) & a \in \agents
        \end{cases}
        \,,
    \]
    in other words, all the agents add their value of~$i$ to how they perceive the bundle~$\bundle_c$ of agent~$c$ (including~$c$); and we update the value~$y_{a,b}$ based on if~$i$ is a best item in the bundle of~$c$ as viewed by the other agents
    \[
        \widehat{y}_{a} =
        \begin{cases}
            y_{a,b} & a \in \agents, b \in \agents \setminus \{c\} \\
            \max \{ y_{a,c}, \val_a(i) \} & a \in \agents
        \end{cases}
        \,.
    \]
    The arcs from the last layer to the sink~$t$ we add from all vertices corresponding to EF1 allocations, that is, if for all pairs of agents~$a,b \in \agents$ the following holds for the vertex~$v_{\bm{x}, \bm{y}}$ in the last layer:
    \[
        x_{a,a} \ge x_{a,b} - y_{a,b} \,,
    \]
    in other words, if we remove one of the most valuable items (with the value~$y_{a,b}$) from~$\bundle_b$ (which is valued by~$a$ as~$x_{a,b}$), then~$a$ does not envy to~$b$ (since~$a$ values their bundle to~$x_{a,a}$).
    Now, it follows from the construction that any~$s$-$t$-path in~$G$ corresponds to allocation of all items to agents that maximizes the social impact and is EF1.
    Furthermore,~$G$ is acyclic, since we always increase the number of the layer (and never go back).
    To finish the argument, we discuss the size of~$G$.
    The number of vertices of~$G$ is
    \begin{align*}
        2 + &|[\nu \cdot \numItems]_0^{\agents \times \agents} \times [\nu]_0 ^{\agents \times \agents} | \cdot \numItems \\
        &=
        2 + \left( (1 + |\nu| \cdot \numItems)^{\numAgents^2} + ((1+\nu)^{\numAgents^2}) \right)\cdot \numItems
    \end{align*}
    a size that is pseudopolynomial if~$\numAgents$ is a constant.
    Note that we can find an~$s$-$t$-path if it exists using, e.g., BFS in linear-time in the size of the graph.
    We stress here that even though the size of the graph~$G$ is clearly impractical even for mild settings of~$\nu, \numAgents, \numItems$, we can generate it in an on-demand fashion and find the solution quickly in practice using, e.g., the A$^*$ algorithm.

    \noindent\textbf{wEF1.}
    If we change the objective from EF1 to wEF1, the only affected part of the construction of the graph~$G$ are the edges towards the sink~$t$.
    This is the only place which is objective-dependent.
    The arcs from the last layer to the sink~$t$ we add from all vertices corresponding to wEF1 allocations, that is, if for all pairs of agents~$a,b \in \agents$ the following holds for the vertex~$v_{\bm{x}, \bm{y}}$ in the last layer:
    \[
        \frac{x_{a,a}}{w_a} \ge \frac{x_{a,b} - y_{a,b}}{w_b} \,,
    \]
    in other words if we remove a best item (with the value~$y_{a,b}$) from~$\bundle_b$ (that is valued by~$a$ a~$x_{a,b}$) and divide by the weight~$w_b$ of~$b$, then~$a$ does not envy to~$b$ (because they value their bundle to~$x_{a,a}$) when devided by the weight~$w_a$ of~$a$.
    Now, it follows from the construction that any~$s$-$t$-path in~$G$ corresponds to allocation of all items to agents that maximizes the social impact and is wEF1.

    \noindent\textbf{sEF1.}
    We change the way in which we update the vector~$\bm{y}$ during the construction of the graph~$G$.
    In this case, we connect the vertex~$v^j_{\bm{x}, \bm{y}}$ with two vertices--$v^j_{\bm{\widehat{x}}, \bm{\widehat{y}}}$ and~$v^j_{\bm{\widehat{x}}, \bm{y}}$.
    Essentially, the two correspond to an additional decision if we decide that~$i$ should or should not be the item we remove from~$a$'s bundle while deciding if the allocation is sEF1 (if not,~$\bm{y}$ stays the same).
    Therefore, we update the value~$y_{a,b}$ as follows
    \[
        \widehat{y}_{a} =
        \begin{cases}
            y_{a,b} & a \in \agents, b \in \agents \setminus \{c\} \\
            \val_a(i) & a \in \agents
        \end{cases}
        \,.
    \]
    The rest is done in the same way as in EF1.

    \noindent\textbf{tEF1.}
    We verify that
    \[
        x_{a,a} + y_{a,b} \ge x_{a,b} - y_{a,b} \,,
    \]
    when adding arcs to~$t$, since this corresponds not only to~$b$ not having the item of~$a$-value~$y_{a,b}$ but at the same time increasing the value of~$a$'s bundle by adding this item to their bundle.

    \noindent\textbf{EF.}
    In this case, we can work with~$V^i_{\bm{x}}$ only, since we do not need to remove any items and therefore it suffices to track the values of the bundles.
    
    \noindent\textbf{EFL.}
    This case is similar to the sEF1 case, that is, we add two arcs corresponding to updating the vector~$\bm{y}$ or not.
    Finally, when adding the arcs to the sink~$t$, we have to check
    \[
        x_{a,a} \ge x_{a,b} - y_{a,b} \quad\text{and}\quad y_{a,b} \le x_{a,a} \,,
    \]
    \[
        x_{a,a} \ge x_{a,b} \,
    \]
    or
    \[
        x_{a,b} = y_{a,b}
    \]
    holds.
\end{proof}

\section{Socially Aware Agents}\label{sec:socially_aware}

Motivated by the experimental work of \citet{HerreinerP2009} and \citet{HosseiniKSVX2025}, which explore what real-life users perceive to be fair outcomes, \citet{Hosseini2024} argues that \emph{``developing fair algorithms may require axioms that are able to capture solutions that take the social context into account beyond perceived envy.''} In this section, we build up on this idea and follow the approach introduced by \citet{FlamminiGV2025} and assume \emph{socially aware} agents. Intuitively, an agent is socially aware if they are willing to accept a formally unfair outcome in case it leads to a greater social impact. 

Formally, we define envy-freeness with socially aware agents as follows.
 
\begin{definition}\label{def:SA-EF}
    We say that an allocation~$\alloc$ is \emph{envy-free with socially aware agents} (SA-EF), if for each pair of agents~$i,j\in\agents$, at least one of the following conditions holds:
    \begin{align*}
        && \val_i(\bundle_i) \geq \val_i(\bundle_j) && \text{or} && s_i(\bundle_j) < s_j(\bundle_j)\,. &&
    \end{align*}
\end{definition}
Based on the above definition, we will say that an agent~$i$ \emph{SA-envies}~$j$, if and only if
$\val_i(\bundle_i) < \val_i(\bundle_j) \text{ and } s_i(\bundle_j) \ge s_j(\bundle_j)$.
Notice that if the social impact of all agents is the same for all agents, SA-EF simply reduces to EF, a computationally notoriously hard fairness notion which is rarely satisfiable. Henceforth, in the rest of the section, we focus on natural relaxations of EF.

\begin{definition}\label{def:SA-EF1}
    We say that an allocation~$\alloc$ is \emph{envy-free up to one item with socially aware agents} (SA-EF1), if for each pair of agents~$i,j\in\agents$, one of the following conditions hold:
    \begin{enumerate}
        \item~$g\in \items$ exists such that~$\val_{i}(\bundle_i) \geq \val_{i}(\bundle_j\setminus\{g\})$, or
        \item~$s_i(\bundle_j) < s_j(\bundle_j)$.
    \end{enumerate}
\end{definition}

Other fairness notions, such as SA-wEF1 and SA-EFL, are then defined analogously to SA-EF1, just with the first condition modified accordingly. Based on the above definition, we will say that an agent~$i$ \emph{SA-envies}~$j$, if and only if
$\val_i(\bundle_i) < \val_i(\bundle_j) \text{ and } s_i(\bundle_j) \ge s_j(\bundle_j)$.

First, we observe that the second condition in the definition of social awareness is meaningful in the sense that it captures the intuition of socially aware agents but is not strong enough to allow for some very unfair solutions.

\begin{observation}
    If we replace~$<$ with~$\leq$ in the second part of the definition of SA-EF (and its relaxations), any SIM allocation is a solution. 
\end{observation}
\begin{proof}
    In any SIM allocation, for each pair of agents~$i$ and~$j$ we have~$\si_i(\bundle_j) \leq \si_j(\bundle_j)$, so the second condition of the definition (with~$\leq$) is always satisfied.
\end{proof}

In their paper, \citet{FlamminiGV2025} proved that there is a modification of the famous envy-cycle elimination algorithm~\cite{LiptonMMS2004} that always finds an SA-EF1 allocation. We extend their result in two ways.
In our first result, we show that SA-swEF1 allocations (and hence also SA-wEF1 allocations) are also guaranteed to exist.
Interestingly, our approach is based on a picking sequence similar to the round-robin algorithm. More precisely, the algorithm is a modification of the weighted picking sequence protocol proposed by \citet{ChakrabortyISZ2021}, which reduces to a round-robin approach in the case of identical agent weights. However, we only allow agents to pick goods for which they maximize social impact, and we force agents to skip their turn whenever they do not maximize the social impact of any unallocated good. The proof then basically follows from \citet{ChakrabortyISZ2021} and a simple observation that agent~$i$ can SA-envy agent~$j$ only if~$j$ chose only goods for which~$i$ maximizes social impact. Therefore, from the point of view of~$i$, we can completely ignore all agents that get an item for which~$i$ does not maximize social impact. 
On the other hand, if we consider an auxiliary instance in which agents only value items for which they maximize impact and get a dummy item that everyone values~$0$ whenever they skip a turn we get an allocation that is swEF1 according to \citet{ChakrabortyISZ2021} and in which: 1) All agents get same bundle plus some dummy items. 2) Every agent values every bundle for which the agent maximizes the social impact and all non-dummy items in it exactly the same in both instance.

\begin{theorem}\label{thm:sa_swef1_goods}
    SIM and SA-swEF1 allocation is guaranteed to exist and can be found in polynomial time.
\end{theorem}
\begin{proof}
    Let us start the proof by describing the weighted picking sequence protocol by \citet{ChakrabortyISZ2021}. Agents pick iteratively items in rounds and we keep for each agent~$i\in\agents$ the integer~$t_i$ representing how many items they picked so far, initialized as~$t_i=0$. In each round, an agent~$i$ that minimizes the value~$\nicefrac{t_i}{w_i}$ picks their preferred item and increases their~$t_i$ by one. \citet{ChakrabortyISZ2021} show that if agent~$i$ picks in some round, then in the partial allocation up to and including~$i$’s latest pick, agent~$j$ is weighted envy-free toward~$i$ up to the first item~$i$ picked. Since allocating goods to agents other than~$i$ cannot increase envy towards~$i$, it follows that every partial allocation, including the final allocation, is swEF1. 

    Let us now modify the picking sequence, such that in each round an agent~$i$ that minimizes the value~$\nicefrac{t_i}{w_i}$ is picked and~$t_i$ is increased by~$1$. However, agent~$i$ is only allowed to pick an item for which they have maximum social impact among all agents and they do not get any item if there is no such item for which they maximize the social impact left. Let~$\sigma$ denote the final picking sequence (that is, in the round~$t$, the agent~$\sigma(t)$ picks an item). Note that the number of rounds can be larger than~$|\items|$, since agents might not be allowed to pick an item. 

    It is easy to see that the final allocation~$\alloc = (\bundle_1, \ldots,\bundle_n)$ is SIM, since for every~$i\in \agents$, the bundle~$\bundle_i$ only contains items for which agent~$i$ has maximum social impact. 
    Moreover, since for all~$i\in\agents$, the social impact function is additive, it follows that if for some pair of agents~$i,j\in\agents$ there is a good~$g\in \bundle_j$ such that~$s_i(g) < s_j(g)$, then~$s_i(\bundle_j) < s_j(\bundle_j)$.
    Let us now argue that~$\alloc$ is also SA-swEF1. Let us fix an agent~$i\in\agents$. As argued above, if there is a good~$g\in \bundle_j$ in the bundle of some agent~$j\in \agents\setminus\{i\}$ such that~$s_i(g) < s_j(g)$, then~$s_i(\bundle_j) < s_j(\bundle_j)$ and the condition of SA-swEF1 from~$i$ towards~$j$ is satisfied. It remains to show that swEF1 condition is satisfied as well for the remaining agents~$j$ that only receive items for which~$i$ maximizes the impact. To do so, we define an auxiliary instance without social impact and with some dummy items, where agent~$j\in\agents$ only values items for which it maximizes the impact and gets as its bundle~$\bundle_j$ plus some dummy items that everyone values~$0$, and we argue that this allocation is swEF1.
    Now, let~$m'\ge |\items|$ be the final number of rounds of the picking sequence~$\sigma$ and let us consider instance~$\mathcal{I}' = (\agents, \items', (v_j')_{j\in\agents'}, (w_j)_{j\in\agents})$ for swEF1 (without social impact) such that 
    \begin{itemize}
        \item~$\items \subseteq \items'$ and~$|\items'| = m'$, 
        \item for all~$j\in \agents$ and all~$g\in \items~$,~$v_j'(g) = v_j(g)$ if~$j$ maximizes impact of~$g$ and~$v_j'(g) = 0$ otherwise, 
        \item for all~$j\in \agents$ and all~$g\in \items'\setminus\items~$,~$v_j(g) = 0$.
    \end{itemize}
    Note that~$\sigma$ is the weighted picking sequence by \citet{ChakrabortyISZ2021} for the instance~$\mathcal{I}'$. Let us show by induction that if in the round~$t$, the agent~$\sigma(t)$ chose item~$g$ in our modified weighted picking sequence protocol, then~$\sigma(t)$ can chose~$g$ according to weighted picking sequence protocol by \citet{ChakrabortyISZ2021} for the instance~$\mathcal{I}'$ and if~$\sigma(t)$ chooses nothing in our protocol, then it can chose a \emph{dummy} item in~$\items'\setminus \items$. Let~$t\in [m']$ and assume that we are following weighted picking sequence protocol by \citet{ChakrabortyISZ2021} for~$\mathcal{I}'$ and in every round~$t'< t$, the agents~$\sigma(t')$ either chose the same item as the agent~$\sigma(t')$ in our modified protocol or a dummy item from~$\items'\setminus \items$. Note, this assumption holds vacuously for~$t=1$. It follows from our assumption that at the beginning of the round~$t$ the set of assigned items from~$\items$ is exactly the same in both protocols and the number of assigned dummy items is exactly the same as the number of times an agent skipped turn. Now, if agent~$\sigma(t)$ skips turn in our modified protocol, then all items for which~$\sigma(t)$ maximizes impact are already allocated. But all remaining items, the agent~$\sigma(t)$ values~$0$ according to~$v_{\sigma(t)}'$. Moreover, by our assumption, not all dummy items have been assigned yet, so~$\sigma(t)$ can choose a dummy item. Otherwise, in our modified protocol,~$\sigma(t)$ chooses an item~$g$ with a maximum value under~$v_{\sigma(t)}$ among the items for which~$\sigma(t)$ maximizes impact. Since under~$v'_{\sigma(t)}$, the agent~$\sigma(t)$ only values items for which it maximizes impact,~$g$ is a maximum valued unassigned item at the moment according to~$v'_{\sigma(t)}$, so weighted picking sequence protocol can assign~$g$ to~$\sigma(t)$ in the round~$t$. 
    This way we obtain an allocation~$\alloc'= (\bundle_1',\ldots, \bundle_n')$ that is swEF1 for the instance~$\mathcal{I}'$ such that for all~$j\in \agents$ it holds that~$\bundle_j\subseteq \bundle_j'$ and~$\bundle_j'\setminus \bundle_j\subseteq \items'\setminus\items$. It then follows that if agent~$i$ maximize the social impact of the bundle of agent~$j$, then for every~$g\in \bundle_j$, we have~$v_i(g) = v_i'(g)$ and for every~$g\in \bundle_j'\setminus \bundle_j$, we have~$v_i'(g) = 0$. Hence, if~$g_1^j$ is the first item that~$j$ picked in the picking sequence, then~$v_i(A_i) = v_i'(A_i') \ge v_i'(A_j'\setminus \{g_1^j\}) = v_i(A_j\setminus \{g_1^j\})$ and so removing the the first item~$j$ picked, removes envy from~$i$ towards~$j$. 
    
    To summarize, for every pair~$i,j\in \agents$, either~$A_j$ contains an item~$g\in\items$ such that~$s_i(g) < s_j(g)$ and so~$s_i(A_j)< s_j(A_j)$ or removing the first item that~$j$ picked removes any potential weighted-envy from~$i$ towards~$j$. Hence, the allocation~$\alloc$ is SA-swEF1.
\end{proof}

The above approach naturally extends to all fairness notions we consider except for SA-EFL. It is not a coincidence, as it turns out that \emph{any} picking sequence protocol that computes picking sequence disregarding the valuation functions is doomed from the start for SA-EFL. Observe that the first agent in the sequence that is allowed to pick more than once can choose an item that everyone, including the agent, considers more valuable than the sum of all remaining items (not assigned until the first pick of the agent); to give a more concrete example, consider the following.

\begin{example}
    Let us assume an instance with two agents, four items, and both agents have identical social impact and identical valuations for these items. Let these identical valuations be~$100$,~$1$,~$1$,~$1$. Consider the round-robin picking sequence, in which each agent picks two items. So, the final allocation is~$(\{100,1\}, \{1,1\})$. Clearly, agent 2 with the bundle~$\{1,1\}$ envies agent 1 with the bundle~$\{100,1\}$, and only removal of~$100$ would remove this envy. However,~$100 > 2 = 1+1$, so this allocation is not EFL. The only way to get an EFL allocation for this instance is to give the large item to one agent and the remaining items to the other.
\end{example}

However, we obtain a positive result for EFL anyway. The algorithm to obtain SA-EFL is a combination of the ideas for SA-EF1 from \citet{FlamminiGV2025} and for EFL from \citet{BarmanBMN2018}. 
The algorithm is basically identical to the algorithm for finding EFL allocation from \citet{BarmanBMN2018}; however, agents are only allowed to pick goods for which they maximize social impact, and instead of using the standard envy-graph, whose vertices are agents and directed edges represent envy from one agent to another, we consider the so-called SA-envy graph, where directed edge represent \emph{SA-envy}. The correctness follows again from the observation that throughout the computation, each agent maximizes social impact of their bundle. So, once a bundle~$\bundle_j$ receives an item for which an agent~$i$ does not maximize the social impact, agent~$i$ will never envy an agent that has (a superset of) the bundle~$\bundle_j$. Hence, considering only the SA-envy graph is sufficient. The proof that the EFL property is satisfied if two agents SA-envy one another is then analogous to the proof of \citet{BarmanBMN2018}.

Before we prove our result for SA-EFL, let us define the SA-envy graph for a partial allocation~$\alloc$ introduced by \citet{FlamminiGV2025}. For a partial allocation~$\alloc$ of items to agents in~$\agents$, we let~$G_{\alloc}$ denote the directed graph whose vertex set is~$\agents$ and there is an arc from agent~$i$ to agent~$j$ in~$E(G_{\alloc})$ if and only if~$\val_i(\bundle_i) < \val_i(\bundle_j) \text{ and } s_i(\bundle_j) \ge s_j(\bundle_j)$; that is~$i$ SA-envies~$j$. An important subprocedure that we need to perform on~$G_{\alloc}$ is so-called \textsc{CycleElimination} (introduced by \citet{LiptonMMS2004}). The input of \textsc{CycleEliminiation} is~$G_{\alloc}$ and a cycle~$\mathcal{C}$ in~$G_{\alloc}$. \textsc{CycleEliminiation} then, for every directed edge from agent~$i$ to agent~$j$ along the cycle, takes the bundle~$\bundle_j$ of agent~$j$ and gives it to agent~$i$. \citet{LiptonMMS2004} observed that \textsc{CycleEliminiation} strictly decreases the number of edges in the envy graph (and the same holds for SA-envy graph by \citet{FlamminiGV2025}), hence after at most~$|E(G_{\alloc})|$ many applications of \textsc{CycleEliminiation} we obtain an acyclic SA-envy graph. 

\begin{theorem}\label{thm:EFL_poly}
    SIM and SA-EFL allocation is guaranteed to exist and can be found in polynomial time.
\end{theorem}
\begin{proof}
    The algorithm works as follows. We initialize the partial allocation~$\alloc$ to be the empty allocation, with every agent receiving no item. Afterwards, we initialize the SA-envy graph~$G_{\alloc}$ as the empty graph with~$\agents$ many vertices. Now, while there is some unassigned good, we pick a source agent~$i$ (an agent without an edge point to it in~$G_{\alloc}$). If there is no good for which~$i$ has maximum social impact, we remove~$i$ from~$G_{\alloc}$ and go to the next source agent. Otherwise, we allocate to the agent~$i$ the most valued good by the agent~$i$ among all the goods for which~$i$ has the maximum social impact. Afterwards, we use \textsc{CycleEliminiation} procedure to eliminate all cycles from~$G_{\alloc}$ introduced by giving the good to agent~$i$ and repeat from picking another source agent in~$G_{\alloc}$. 
    
    Since the algorithm is nearly identical to the one from \citet{BarmanBMN2018}, but we consider the SA-envy graph (a subgraph of the envy graph) and we only allow agents to pick from a subset of items. It is straightforward to see that the algorithm runs in polynomial time. More precisely, the algorithm needs to 1) find a source vertex in an acyclic graph at most~$|\agents|+|\items|$ times, 2) remove a vertex from the graph at most~$|\agents|$ many times, and 3) find a cycle and apply \textsc{CycleEliminiation} at most~$|\items|\cdot |\agents|^2$ many times. All of these operations can be easily performed in polynomial time. 
    
    Now, we will show by induction that every partial allocation throughout the computation is SIM and SA-EFL. This clearly holds for the empty allocation. Now let~$\alloc = (\bundle_1, \ldots, \bundle_n)$ be a SIM and SA-EFL partial allocation during our computation. We consider two possibilities to obtain the consecutive partial allocation. 
    
    First, agent~$i\in\agents$ picks a good~$g$. Since ~$\alloc$ is SIM and~$i$ has a maximum social impact for~$g$, all agents have maximum social impact for all items in their respective bundles, and the allocation is SIM. Moreover, as we only consider goods,~$v_i(\bundle_i\cup g)\ge v_i(A_i)$ and hence the agent~$i$ satisfies the definition of SA-EFL towards any other agent~$j\in \agents\setminus\{i\}$. Now, every agent~$j\in \agents$ that has been removed from~$G_{\alloc}$ does not have maximum impact on~$g$ and hence also on the bundle~$\bundle_i\cup \{g\}$, so such an agent~$j$ cannot SA-envy agent~$i$. Agent~$i$ is a source vertex in~$G_{\alloc}$, so any agent~$j\in V(G_{\alloc})$ does not SA-envy bundle~$\bundle_i$. If agent~$j$ starts to SA-envy~$\bundle_i\cup \{g\}$, then~$j$ has maximum social impact for the good~$g$. If~$g$ is the only good in~$\bundle_i\cup \{g\}$, then SA-EFL from~$j$ towards~$i$ is also satisfied. Otherwise, because~$j$ SA-envies~$\bundle_i\cup \{g\}$, but not~$\bundle_i$, we get that~$s_j(A_i) = s_i(A_i)$ and~$v_j(A_j) \ge v_j(A_i) > 0$, so~$A_j$ is not empty. Since an empty bundle cannot be part of an SA-envy cycle, it follows that agent~$j$ has already picked an item~$g'$ in some previous round. Since~$j$ was allowed to pick~$g$ at that point as well, it follows that~$v_j(g')\ge v_j(g)$. But both operations, 1) allocating a good to an agent and 2) the \textsc{CycleElimination}, do not decrease the value of an agent's bundle. Hence,~$v_j(\bundle_j)\ge v_j(g')\ge v_j(g)$ and~$v_j(\bundle_j)\ge v_j(\bundle_i) = v_j((\bundle_i\cup\{g\})\setminus\{g\})$, and~$j$ satisfy the condition of SA-EFL towards~$i$. 
    
    Second, we apply \textsc{CycleElimination} on a cycle~$\mathcal{C}$ in~$G_{\alloc}$ to obtain allocation~$\alloc'$. Since~$\alloc$ is SIM, an agent~$i$ can only SA-envy a bundle for which~$i$ maximizes social impact. So, if there is an edge from~$i$ to~$j$ on~$\mathcal{C}$, then~$i$ also maximizes social impact of all items in~$\bundle_j$, hence after \textsc{CycleElimination}, the allocation is SIM. Now consider any two agents~$i$ and~$j$. After \textsc{CycleElimination}, agent~$i$ has a bundle~$\bundle_{i'}$, for~$i'\in \agents$ such that~$v_{i}(\bundle_{i'})\ge v_{i}(\bundle_i)$ and agent~$j$ has a bundle~$\alloc_{j'}$, for~$j'\in\agents$. By our inductive assumption, we have that~$\alloc$ is SA-EFL. Hence, agent~$i$ satisfies SA-EFL condition toward agent~$j'$. That is either~$s_i(A_{j'}) < s_{j'}(A_{j'})$ or there is a good~$g\in A_{j'}$ such that~$v_i(g)\le v_i(A_i)$ and~$v_i(A_i)\ge v_i(A_{j'}\setminus \{g\})$.
    But then,~$s_{j}(A_{j'}) = s_{j'}(A_j)$ and~$v_i(A_{i'})\ge v_i(A_i)$, so we get either~$s_i(A_{j'}) < s_{j}(A_{j'})$ or there is a good~$g\in A_{j'}$ such that~$v_i(g)\le v_i(A_{i'})$ and~$v_i(A_{i'})\ge v_i(A_{j'}\setminus \{g\})$. In other words, agent~$i$ satisfies the SA-EFL condition towards agent~$j$.
    
    It follows that every partial allocation, including the final allocation of all items, is SIM and SA-EFL.
\end{proof}

Previous results indicate that whenever we combine a fairness notion with social awareness, we obtain an existence guarantee accompanied by polynomial-time algorithms. One may wonder whether these very positive results are not caused simply by the social awareness of our agents. Maybe it is the case that whenever we have socially aware agents, we can actually achieve an arbitrary fairness notion. To formally study this possibility, we assume fully envious but socially aware agents. An agent~$i\in\agents$ is not satisfied with an allocation~$\alloc$ whenever another agent~$j$ exists with a non-empty bundle. Formally, the notion is defined as follows.

\begin{definition}\label{def:czechEnvy}
    We say that an allocation~$\alloc$ is \emph{SA-$\emptyset$}, if for each pair of agents~$i,j\in\agents$ at least one of the following conditions holds:
    \begin{align*}
        && \bundle_j = \emptyset && \text{or} && s_i(\bundle_j) < s_j(\bundle_j)\,. &&
    \end{align*}
\end{definition}

First, we present an auxiliary lemma that significantly restricts the space of possible allocations and is helpful for multiple subsequent results.

\begin{lemma}\label{thm:czechEnvy:sameAgentTypesEmpty}
    If there are two agents~$a_i$ and~$a_{i'}$ of the same agent-type, then in every SA-$\emptyset$ allocation~$\alloc$ it holds that~$\bundle_i = \bundle_{i'} = \emptyset$.
\end{lemma}
\begin{proof}
    For the sake of contradiction, assume that without loss of generality,~$\bundle_i \not= \emptyset$. Then, it must be the case that~$s_{a_i}(\bundle_i) > s_{a}(\bundle_i)$ for every~$a\in\agents\setminus\{a_i\}$; otherwise~$\alloc$ is not SA-$\emptyset$. However,~$a_{i'}$ is of the same agent-type as~$a_i$, meaning that for every set~$S\subseteq \items$~$s_{a_i}(S) = s_{a_{i'}}(S)$. Hence,~$\alloc$ is not SA-$\emptyset$, so the bundles of both such agents must be empty.
\end{proof}

However, as the next theorem proves, even though we can focus only on a very specific subset of allocations, the problem of deciding the existence of SIM and SA-$\emptyset$ allocations is computationally hard. That is, socially aware agents are not enough to guarantee arbitrary fairness notion.

\begin{theorem}\label{thm:czechEnvy:NPh}
    It is \NPc to decide whether there is a SIM and SA-$\emptyset$ allocation, even if the valuations and social impacts are binary and~$v_i = v_j$ for every~$i,j\in\agents$.
\end{theorem}
\begin{proof}
    The reduction is from the \probName{Restricted Exact Cover by~$3$-Sets} problem, which is known to be \NPc~\cite{Gonzalez1985}. In this problem, we are given a universe~$\mathcal{U}=(u_1,\ldots,u_{3\ell})$ and a family of subsets~$\mathcal{S} = (S_1,\ldots,S_{3\ell})$ such that (a)~$|S_j| = 3$ for every~$j\in[3\ell]$, (b) each~$u_i\in\mathcal{U}$ is element of exactly three sets of~$\mathcal{S}$, and (c) each pair of sets~$S_j$ and~$S_{j'}$ intersects in at most one element. The goal is to decide if there exists a set~$X\subseteq \mathcal{S}$ such that~$|X| = \ell$ and~$\bigcup_{S\in X} S = \mathcal{U}$.

    Given an instance~$\mathcal{I}=(\mathcal{U},\mathcal{S})$, we construct an equivalent instance~$\mathcal{J}$ of our problem as follows (see \Cref{fig:czechEnvy:construction} for an illustration of the construction). We have~$3\ell + 2$ agents~$a_1,\ldots,a_{3\ell},a_{3\ell+1},a_{3\ell+2}$. Agents~$a_1,\ldots,a_\ell$ are called \emph{set agents} and are in one-to-one correspondence to the sets of~$\mathcal{S}$. Agents~$a_{3\ell+1}$ and~$a_{3\ell+2}$ are special \emph{guard agents}. The set of items contains one \emph{element item}~$g_i$ for every~$u_i\in\mathcal{U}$ together with~$\ell$ special \emph{dummy items}~$d_1,\ldots,d_\ell$. The valuation function for every agent assigns a value of~$1$ to every item. The important bit is the social impact functions. The guard agents are identical and have a social impact of~$1$ for every element item. Otherwise, their social impact is zero. Each set agent~$a_j$,~$j\in[3\ell]$, on the other hand, has a social impact of~$1$ for every dummy item and each~$g_i$ such that~$u_i\in S_j$. For all the remaining items, the social impact is zero. Observe that each set agent maximizes social impact for exactly~$3$ element items, and each element item can be assigned to exactly five agents---three set agents and the guard agents.

    \begin{figure}[bt!]
        \centering
        \renewcommand{\arraystretch}{1.2}
        \begin{tabular}{c|cccc|ccc}
                          &~$g_1$ &~$g_2$ &~$\cdots$ &~$g_{3\ell}$ &~$x_1$ &~$\cdots$ &~$x_\ell$ \\\hline
           ~$a_1$         &   1   &   1   &          &       1     &   1   &~$\cdots$ &     1     \\
           ~$a_2$         &   1   &   0   &          &       1     &   1   &~$\cdots$ &     1     \\
           ~$\vdots$      &       &       &          &             &       &          &           \\
           ~$a_{3\ell}$   &   0   &   1   &          &       0     &   1   &~$\cdots$ &     1    \\\hline
           ~$a_{3\ell+1}$ &   1   &   1   &~$\cdots$ &       1     &   0   &~$\cdots$ &     0     \\
           ~$a_{3\ell+2}$ &   1   &   1   &~$\cdots$ &       1     &   0   &~$\cdots$ &     0     \\
        \end{tabular}
        \caption{An illustration of the social impact function used in the instance constructed in the proof of \Cref{thm:czechEnvy:NPh}.}
        \Description{An illustration of the construction used to prove \Cref{thm:czechEnvy:NPh}.}
        \label{fig:czechEnvy:construction}
    \end{figure}

    The crucial idea behind the construction is that whenever an agent~$i$'s bundle is non-empty, then for any other agent~$i'$, there exists an item~$g\in \bundle_i$ such that~$s_{i}(g) > s_{i'}(g)$. Consequently, if any guard agent's bundle is non-empty, then he is envious of the other guard agent, and whenever a set agent's bundle is non-empty, it also has to contain at least one dummy item, as otherwise, she is envious of guard agents.

    For correctness, assume first that~$\mathcal{I}$ is a \Yes-instance and~$X\subseteq \mathcal{S}$ is an exact cover of~$\mathcal{U}$. We create an allocation~$\alloc$ as follows. The bundles of the guard agents and every set agent~$a_j$ such that~$S_j\not\in X$ are empty, and therefore, they never envy any other agent. For every~$a_j$ such that~$S_j\in X$, we add to her bundle all element items corresponding to elements of~$S_j$ and an arbitrary not yet allocated dummy item. Note that since~$X$ is a solution, the bundles are well defined as there are~$\ell$ dummy items and~$|X|=\ell$. It remains to show that~$\alloc$ is indeed a solution for~$\mathcal{J}$. Let~$a_i$ and~$a_{i'}$ be a pair of agents such that~$i$ is envious of~$i'$s bundle. By definition,~$\bundle_i\not=\emptyset$; thus,~$a_i$ is a set agent. Since~$\bundle_i$ contains some dummy item,~$a_{i'}$ is not a guard agent. Hence, it must be the case that~$s_{a_i}(\bundle_i) = s_{a_{i'}}(\bundle_i)$, which implies that sets~$S_{i}$ and~$S_{i'}$ are identical; however, this is not possible by our assumption (c). That is, such a pair of agents cannot exist, and we obtain that~$\alloc$ is indeed a solution.

    In the opposite direction, let~$\alloc$ be a SA-$\emptyset$ allocation for~$\mathcal{J}$. We show several properties of~$A$, which altogether implies that~$\mathcal{I}$ is also a \Yes-instance. First, since the guard agents are identical, by \Cref{thm:czechEnvy:sameAgentTypesEmpty} their bundles are necessarily empty. Now, let~$a_i$ be a set agent with a non-empty bundle and assume that~$\bundle_i$ does not contain any dummy item. Then, the set agent has envy towards both guard agents, as they are social impact maximizers on all element items. That is the bundle of every set agent~$a_i$ such that~$\bundle_i\not= \emptyset$ contains at least one dummy item. Finally, let there be a set agent~$a_i$ with~$\bundle_i\not = \emptyset$ such that~$\bundle_i$ contains at least two dummy items. Then, there are at most~$\ell-2$ other set agents with non-empty sets; however, each of them can have at most three element items, meaning that at most~$(\ell-2)\cdot 3 + 3 = 3\ell - 6 + 3 = 3\ell - 3$ are distributed between these agents. However, since the bundles of guard agents and all remaining set agents are empty, it follows that~$\alloc$ is not a complete allocation, which is not possible. Consequently, no~$a_i$ has assigned more than one dummy item. If we take~$X = \{ S_j \mid \bundle_{a_j}\not=\emptyset \}$, we obtain an exact cover of~$\mathcal{U}$ by the previous argumentation.
\end{proof}

We conclude with one positive algorithmic result for \mbox{SA-$\emptyset$}. Specifically, we show that whenever the number of agent-types or item-types is bounded, there is an efficient algorithm deciding the existence of SIM and SA-$\emptyset$ allocations. This result, especially for the latter parameterization, is not without interest, as in the area of fair division of indivisible items, parameterization by the number of different item-types is a notoriously hard open problem; see, e.g., \cite{EibenGHO2023,NguyenR2023,BredereckKKN2023}.

\begin{theorem}
    When parameterized by the number of item-types or the number of agent-types, an \FPT algorithm deciding the existence of SIM and SA-$\emptyset$ allocation exist.
\end{theorem}
\begin{proof}
    First, assume that the number of agent-types~$t_\agents$ is our parameter. By definition, two items are of the same type if the set of agents that are social impact maximizers for them are the same. Similarly, two agents are of the same agent-type if they are social impact maximizers for the set of items. Consequently, whenever the number of agent-types is bounded, the number of items-types is at most~$2^\Oh{t_\agents}$. Therefore, in the rest of the proof, we focus only on the parameterization by the number of item-types.

    Let~$t_\items$ be the number of different item-types. We first partition our items into sets~$T_1,\ldots,T_{t_\items}$ according to their types. Additionally, let~$\agents_t$ be the set of agents maximizing social impact for items of type~$t\in[t_\items]$. Next, we observe that whenever the number of item-types is bounded, so is the number of agent-types; specifically,~$t_\agents \in 2^\Oh{t_\items}$. Moreover, by \Cref{thm:czechEnvy:sameAgentTypesEmpty}, if there are at least two agents that are of the same agent-type, then bundles of all agents of this agent-type are empty. Hence, the number of agents that are of unique agent-type is also~$2^\Oh{t_\items}$, and we use~$U$ to denote the set of all such ``unique'' agents. Moreover, we guess the subset~$U'\subseteq U$ of such unique agents whose bundles are non-empty in the solution.

    To verify whether an SA-$\emptyset$ allocation exists, we formalize the problem as an ILP with fixed dimension. We have a non-negative variable~$x_{i,t}$ for every~$i\in U'$ and every~$t\in[t_\items]$ such that~$i\in T_t$ representing the number of items of type~$t$ in the bundle of agent~$i$. Our first set of constraints ensures that each item is allocated to some agent of~$U$. To ensure this, we add
    \begin{align}
        && \forall t\in[t_\items] && \sum_{i\in U'\cap \agents_t} x_{i,t} = |T_t|\,. &&
    \end{align}
    Next, for every agent~$i\in U'$, we add a single constraint to ensure that their bundle is indeed non-empty
    \begin{align}
        && \forall i\in U' && \sum_{t\colon i\in \agents_t} x_{i,t} \geq 1 &&
    \end{align}
    Finally, the rest of the constraints ensure that if an agent has a non-empty bundle, then there is no agent with the same social impact on the same set of items. Specifically, we add the following constraints.
    \begin{align}
        \forall i\in U'~\forall j \in \agents\setminus\{i\} && \sum_{t\colon i\in \agents_t} x_{i,t} > \sum_{t'\colon j\in \agents_{t'}} x_{i,t'}\label{eq:czechEnvy:FPT:SAcond}
    \end{align}
    Observe that since the agents of the same agent-type are maximizing social impact for exactly the same set of items, we can create only~$2^\Oh{t_\items}$ constraints of type \eqref{eq:czechEnvy:FPT:SAcond} for every~${i\in U'}$.

    The number of variables of this ILP is~$t_\items \cdot 2^\Oh{t_\items}\in 2^\Oh{t_\items}$. Therefore, we can use the celebrated theorem of Lenstra Jr.~\citeyearpar{Lenstra1983} to solve such an ILP in time~$(2^{t_\items})^{2^\Oh{t_\items}}\cdot (\numAgents+\numItems)^\Oh{1} \in 2^{2^\Oh{t_\items}}\cdot (\numAgents+\numItems)^\Oh{1}$ time. Moreover, there are~$2^{2^\Oh{t_\items}}$ different choices for~$U'$, so the overall running time of our approach is~$2^{2^\Oh{t_\items}}\cdot (\numAgents+\numItems)^\Oh{1}$, which is clearly in \FPT.
\end{proof}

\section{Limited Social Awareness}\label{sec:relaxedSA}

In the previous section, we studied the impact of social awareness on the existence of fair and social impact maximizing allocations. It turned out that with socially aware agents, such a desirable allocation always exists and can be found efficiently.
However, the notion of social awareness, as defined, is somewhat idealistic, and in practical scenarios, it may not be the case that agents accept very bad bundles just because of the social impact they generate with these chores. Or, some agents can be even very selfish and may be interested only in fairness and not in the social impact at all.

\subsection{Some Agents Socially Unaware}

The first limitation of social awareness we study is a setting with some agents socially aware and some socially unaware. This is arguably a very natural setting. Unfortunately, as we show in the following example, even if we allow as few as one agent to be socially unaware, we lose all existence guarantees.

\begin{example}
    Let~$a_1$ be a socially unaware agent,~$a_2$ be socially aware agent, and~$g_1,g_2$ be two identical elements such that for~$j\in[2]$ we have~$\val_{a_1}(g_j) = \val_{a_2}(g_j) = 10$, and~$s_{a_1}(g_j) = 0$ and~$s_{a_2}(g_j) = 1$. Due to the social impact function, both items have to be allocated to agent~$a_2$. However,~$a_1$ is now very envious of agent~$a_2$, and this envy cannot be eradicated by removing any item from~$\bundle_1$.
\end{example}

Note that we can extend the previous example to show the non-existence of the SA-EF$r$ solution for arbitrary~$r\in\N$, simply by adding at least~$r+1$ copies of our items. Hence, even if we have exactly one socially unaware agent, we cannot provide any guarantee on (relaxations of) EF.

In the next result, we draw the situation even more negatively. Specifically, we show that even one socially unaware agent~$a_1$ makes the decision of the existence of a social impact maximizing allocation, which is EF1 for~$a_1$ and SA-EF1 for all other agents, computationally hard. 

\begin{theorem}\label{thm:oneSociallyUnaware:NPh}
    It is \NPc to decide if a SIM allocation which is EF1 for agent~$a_1$ and SA-EF1 for every~$a_i\in\agents\setminus\{a_1\}$ exist, even if there are only two agents.
\end{theorem}
\begin{proof}
    We reduce from the \probName{Partition} problem. Recall that in this problem, we are given a multi-set of integers~$W=(w_1,\ldots,w_\ell)$ such that~$\sum_{j\in[\ell]} w_j = 2t$, and the goal is to decide whether~$J\subseteq[\ell]$ exists such that~$\sum_{j\in J} w_j = \sum_{j\in [\ell]\setminus J} w_j = t$. This problem is known to be (weakly) \NPh~\cite{GareyJ1979}. Without loss of generality, we can assume that no~$w_j\in W$ is of size~$\geq t$, as otherwise we have an obvious \Yes (if~$w_j = t$) or \No ($w_j > t$) instance.

     Given an instance~$\mathcal{I}$ of \probName{Partition}, we construct the following equivalent instance of our problem as follows. We have two agents,~$a_1$ and~$a_2$, where~$a_2$ is the only socially aware agent. Next, we have an item~$g_j$ for every element~$w_j\in W$. Both agents value this item as~$w_j$ and are maximizing social impact for these items; that is,~$s_{a_i}(g_j) = 1$ for every~$i\in[2]$. 
     Additionally, we add three ``big'' items~$G_1,G_2$ and~$G_3$. 
     The first agent values~$G_1$ and~$G_2$ at~$t$, while has zero social impact for them, and values~$G_3$ at 0, while has social impact equal to 1 for it.
     Conversely, the second agent values~$G_1$ and~$G_2$ at~$0$,~$G_3$ at~$t$, and has social impact equal to 1 for each one of them.
     The construction is summarized in the following table.
    
    \begin{center}
        \renewcommand{\arraystretch}{1.2}
        \begin{tabular}{c|c|c|c}
                  &~$G_1,G_2$     &~$G_3$         &~$g_j$,~$j\in[\ell]$  \\\hline
           ~$a_1$ &~$s=0, \val=t$ &~$s=1, \val=t$ &~$s=1, \val=w_j$\\
           ~$a_2$ &~$s=1, \val=0$ &~$s=1, \val=t$ &~$s=1, \val=w_j$\\
        \end{tabular}
    \end{center}

    For the correctness, let~$J\subseteq[\ell]$ be a solution for~$\mathcal{I}$. We define the bundles as follows. We set~$\bundle_1 = \{G_3\} \cup \{ g_i \mid i \in J \}$ and~$\bundle_2 = \{G_1,G_2\} \cup \{ g_i \mid i \in [\ell]\setminus J \}$. Recall that agent~$a_1$ is not socially aware and~$\val_{a_1}(\bundle_1) = 2t$, while~$\val_{a_1}(\bundle_2) = 3t$. That is, agent~$a_1$ is envious of~$a_2$'s bundle; however, this envy can be removed if~$G_1$ is deleted from~$a_2$'s bundle. 
    For agent~$a_2$, on the other hand, we have that~$\val_{a_2}(\bundle_2) = t$, while~$\val_{a_2}(\bundle_1) = 2t$. 
    Nevertheless, if we remove~$G_3$ from~$a_1$'s bundle, we obtain that~$\val_{a_2}(\bundle_1\setminus\{G_3\}) = 2t -t = t$. 
    In addition, observe that the allocation indeed maximizes the social impact.
    That is, the allocation~$\alloc$ is EF1 for both agents.

    In the opposite direction, let~$\alloc$ be an allocation which is EF1 for~$a_1$ and SA-EF1 for~$a_2$. 
    First we will prove that in any solution~$G_3 \in \bundle_1$. Then, we  will prove that indeed any solution to our problem gives a solution to the original \probName{Partition} instance.
    Since we are interested in social impact maximizing allocation, it must hold that~$\{G_1,G_2\} \subseteq \bundle_2$.
    
    For the sake of contradiction assume that ~$G_3 \in \bundle_2$.
    Then, agent~$a_1$ values~$\bundle_2$ at least~$3t$.
    If~$\val_1(\bundle_2) > 3t$, then~$\val_1(\bundle_2\setminus \{g\}) > 2t$ for every~$g \in \bundle_2$. Hence, we cannot eliminate envy for agent 1 in this case. 
    Thus, it must be true that~$\val_{a_1}(\bundle_2) = 3t$. This implies that all remaining items belong to~$\bundle_1$. 
    But in this case we have that~$\val_{a_2}(\bundle_1)=2t$, thus~$\val_2(\bundle_1\setminus \{g\}) > t$, for every~$g$, and~$s_{a_2}(\bundle_1) = s_{a_1}(\bundle_1)$. Hence, SA-EF1 criterion is not satisfied for the socially aware agent~$a_2$. Hence, we conclude that~$G_3 \in \bundle_2$.

    Let now ~$\bundle_1 = J \cup \{G_3\}$ and~$\bundle_2 = \bar{J} \cup \{G_1, G_2\}$, where~$(J, \bar{J})$ is a partition of the non-big items.
    We claim that it must be true that in any acceptable solution it must be true that~$\val_{a_1}(\bundle_1\setminus\{G_3\}) = \val_{a_1}(\bundle_2 \setminus \{G_1,G_2\}) = \val_{a_2}(\bundle_1\setminus\{G_3\}) = \val_{a_2}(\bundle_2 \setminus \{G_1,G_2\}) = t$. In other words,~$J$ is a valid solution for the \probName{Partition} instance.
    For the sake of contradiction, assume that~$\val_{a_2}(\bundle_2) = \val_{a_2}(\bundle_2 \setminus \{G_1,G_2\}) = 2t -\delta$ for some~$\delta > 0$. 
    Observe that~$\val_{a_2}(\bundle_2) = \val_{a_2}(\bar{J})$. 
    In addition, observe that 
    \begin{align*}
        \val_{a_2}(\bundle_1) 
        & = \val_{a_2}(G_2) + \val_{a_2}(J)\\
        & = t + 2t - \val_{a_2}(\bar{J})\\
        & > 2t + \delta.
    \end{align*}
    Thus,~$\val_{a_2}(\bundle_1 \setminus \{g\}) > 2t > \val_{a_2}(\bundle_2)$. In addition, observe that~$s_{a_2}(\bundle_1) = s_{a_2}(\bundle_1)$. Thus, the conditions of SA-EF1 are not satisifed for agent~$a_2$. Hence, we conclude that it cannot be true that~$\val_{a_2}(\bundle_2) < t$, in any solution of our problem. Hence, it must hold that~$\val_{a_2}(\bundle_2) \geq t$.

    Suppose now, again for the sake of contradiction, that~$\val_{a_2}(\bundle_2) = t + \delta$, for some~$\delta > 0$. 
    Thus, we get that ~$\val_{a_2}(\bundle_2 \setminus \{G_1,G_2\})= \val_{a_2}(\bar{J}) = t + \delta$. Recall, by construction~$\val_{a_2}(\bar{J}) = \val_{a_1}(\bar{J})$ for every set~$J$.
    Then, we get that~$\val_{a_1}(\bundle_2) = \val_{a_1}(\{G_1,G_2\}) + \val_{a_1}(\bar{J}) = 3t +\delta$. 
    Moreover,~$\val_{a_1}(\bundle_1) = \val_{a_1}(G_3) + \val_{a_1}(J) = 2t -\delta$. Hence,~$\val_{a_1}(\bundle_2 \setminus \{g\}) > 2t + \delta > \val_{a_1}(\bundle_1)$. Hence, the EF1 criterion for agent~$a_1$ is not satisfied. A contradiction. Hence, we can conclude that~$\val_{a_2}(\bundle_2) = t$, which in turn means that~$\val_{a_2}(\bar{J}) = \val_{a_2}(J) = t$, i.e.,~$J$ is a solution for the initial instance of \probName{Partition}, which completes the proof.
\end{proof}

In the previous hardness result, we exploited the fact that some items are forced to be allocated to the not socially aware agent. In the following, we show that in the case of two agents, this property is crucial to draw an instance computationally hard because if there is no such item, we can decide the instance in polynomial time.

\begin{proposition}
    If there are two agents and no item where~$a_2$ is the unique agent maximizing social impact, an allocation that is simultaneously SIM, EF1 for agent~$a_1$, and SA-EF1 for agent~$a_2$ always exists and can be found in poly-time.
\end{proposition}
\begin{proof}
    If all items can go to both agents, i.e., social impact functions are identical, the question reduces to deciding whether an EF1 allocation exists. It is known that, e.g., by the round-robin mechanism, EF1 allocation always exists and can be found in polynomial time. Otherwise, there is at least one item~$g$ in which the socially unaware agent~$a_1$ is the unique agent that maximizes social impact. In this case, we allocate \emph{all} items to agent~$a_1$. Even though the agent~$a_2$ is very envious of~$a_1$, it holds that~$s_{a_2}(\bundle_{a_1}) < s_{a_1}(\bundle_{a_1})$ because of~$g$. Therefore, there is no SA-EF1-envy from~$a_2$ to~$a_1$, so such an allocation is indeed a solution.
\end{proof}

We conclude this subsection with a pseudo-polynomial algorithm similar to the one of \Cref{UMandEF1:pseudopolynomialConstantNumberOfAgents}, complementing the hardness from \Cref{thm:oneSociallyUnaware:NPh}.

\begin{theorem}\label{thm:socially_Unaware_pseudopoly}
    For every fixed number of agents~$\numAgents$ and any~${c \leq \numAgents}$, there exists a pseudo-polynomial time algorithm deciding whether a SIM allocation, which is simultaneously~$\mathcal{F}$-fair for~$a_i$,~$i\in[c]$, and SA-$\mathcal{F}$-fair for any~$a_i$,~$i\in[c+1,\ldots,\numAgents]$, exists for any~$\mathcal{F}\in\{\text{tEF1},\text{EF1},\text{sEF1},\text{wEF1},\text{swEF1},\text{EFL}\}$.
\end{theorem}
\begin{proof}
    The algorithm is the same as in \Cref{UMandEF1:pseudopolynomialConstantNumberOfAgents}, we just add for each (ordered) pair of agents~$a$ and~$b$ an indication whether~$b$ contains an item~$g$ such that~$\si_b(g) > \si_a(g)$. If this indication is true, then~$a$ can no longer envy agent~$b$ and we adapt the edges of the graph accordingly.
\end{proof}

\subsection{Approximate Awareness}

A clear outcome of the previous subsection is that having \emph{all} agents socially aware is a necessary condition for the problem to be tractable. However, having fully socially aware agents seems highly unrealistic. Suppose, for example, an instance with two agents~$a_1$ and~$a_2$, both socially aware, and one item~$g$ valued as~$1000$ by both agents and with~$s_{a_1}(g) = 1$ and~$s_{a_2}(g) = 0$. By the social impact functions, we have to allocate the single item to~$a_1$. Observe that~$a_2$ values its bundle at~$0$ and~$a_1$'s bundle at~$1000$. By \Cref{def:SA-EF1}, this tremendous envy is outweighed by the social impact, which is, however, larger only by one. Motivated by this obvious overestimation of the social awareness, in the rest of this section we study several relaxations of SA.

The first approach we explore is with agents who are not fully socially aware but require their social impact to be significantly larger compared to the social impact of every other agent with the same bundle in order to accept an unfair allocation. We define the notion as follows.

\begin{definition}[$\alpha$-SA-EF1]\label{def:alphaSA}
    We say that an allocation~$\alloc$ is \emph{$\alpha$-envy-free up to one item with socially aware agents} ($\alpha$-SA-EF1) for~$0 \leq \alpha \leq 1$, if for each pair of agents~$i,j\in\agents$, one of the following conditions hold:
    \begin{enumerate}
        \item there is no EF1-envy from~$i$ towards~$j$, or
        \item~$s_i(\bundle_j) < \alpha\cdot s_j(\bundle_j)$.
    \end{enumerate}
\end{definition}

It is easy to see that~$1$-SA-EF1 is equivalent to SA-EF1 and that~$0$-SA-EF1 is equivalent to the standard EF1. Therefore,~$1$-SA-EF1 allocations are guaranteed to exist and can be found in polynomial time, and deciding the existence of~$0$-SA-EF1 is computationally hard. Any value of~$\alpha$ between these two extreme values then naturally limits social awareness of our agents.
In our first result, we show that for \emph{any}~$\alpha$-relaxation,~${\alpha < 1}$, of social awareness, unfortunately, the SIM and~$\alpha$-SA-EF1 allocations are not guaranteed to exist.

\begin{proposition}\label{thm:alphaSA:notGuaranteed}
    For any~$\alpha\in\mathbb{R}$,~$0 \leq \alpha < 1$, SIM and~$\alpha$-SA-EF1 is not guaranteed to exist.     
\end{proposition}
\begin{proof}
    Let~$\alpha$ be as stated. We construct an instance with two agents~$a_1$ and~$a_2$ and two items~$g_1$ and~$g_2$. The valuation functions are identical and assign the same value of~$1$ to both items. The social impact functions are
    \[
        s_{a_1}(g_j) = 1 \qquad\text{and}\qquad s_{a_2}(g_j) = \alpha\,.
    \]
    Due to the social impact function, both items have to be allocated to agent~$a_1$. Moreover,~$\val_{a_2}(\bundle_1) = 2$ and~${\val_{a_2}(\bundle_2) = 0}$, meaning that agent~$a_2$ is EF1-envious towards~$a_1$. The social impacts are~$s_{a_1}(\bundle_1) = 2$ and~$s_{a_2}(\bundle_1) = 2\alpha$. That is, the second condition from the definition of~$\alpha$-SA-EF1 is not met.%
\end{proof}

Moreover, the following result shows that relaxing social awareness not only eliminates any existence guarantee, but also renders the associated decision problem computationally intractable. The reduction is similar to the one used to prove \Cref{thm:oneSociallyUnaware:NPh}.

\begin{theorem}\label{thm:alphaSA:NPc}
    For any~$\alpha\in \mathbb{R}$,~$0 \leq \alpha < 1$, it is \NPc to decide whether there is a SIM and~$\alpha$-SA-EF1 allocation, even if there are only two agents.
\end{theorem}
\begin{proof}
    \newcommand{\siBigAOne}{\ensuremath{x}}
    \newcommand{\siBigATwo}{\ensuremath{y}}
    We reduce from the \probName{Partition} problem. Recall that in this problem, we are given a multi-set of integers~$W=(w_1,\ldots,w_\ell)$ such that~$\sum_{j\in[\ell]} w_j = 2t$, and the goal is to decide whether~$J\subseteq[\ell]$ exists such that~$\sum_{j\in J} w_j = \sum_{j\in [\ell]\setminus J} w_j = t$.
    This problem is known to be (weakly) \NPh~\cite{GareyJ1979}.
    Without loss of generality, we can assume that no~$w_j\in W$ is of size~$\geq t$, as otherwise we have an obvious \Yes (if~$w_j = t$) or \No ($w_j > t$) instance.

    Given an instance~$\mathcal{I}$ of \probName{Partition}, we construct the following equivalent instance of our problem as follows.
    We have two agents~$a_1$ and~$a_2$.
    Next, we have an item~$g_j$ for every element~$w_j\in W$.
    Both agents value this item as~$w_j$ and are maximizing social impact for these items; that is,~$s_{a_i}(g_j) = 1$ for every~$i\in[2]$.
    Additionally, we add three ``big'' items~$G_1$,~$G_2$, and~$G_3$.
    Agent~$a_1$ values~$G_1,G_2$ as~$0$, while he values~$G_3$ as~$t$.
    Agent~$a_2$ values all three big items~$G_1,G_2,G_3$ as~$t$.
    Before we define the social impact, we define
    \begin{equation*}\label{eq:setting_siBigAOne}
        \siBigAOne = 2 \cdot \nicefrac{\ell}{\alpha}
        \quad\text{and}\quad
        \siBigATwo = \nicefrac{\ell}{\alpha}\,.
    \end{equation*}
    Then, the social impact functions are as follows~$\si_{a_1}(G_1) = \si_{a_1}(G_2) = \siBigAOne$,~$\si_{a_2}(G_1) = \si_{a_2}(G_2) = \si_{a_1}(G_3) = 0$, and~$\si_{a_2}(G_3) = y$.
    The construction is summarized in the following table.
    
    \begin{center}
        \renewcommand{\arraystretch}{1.2}
        \begin{tabular}{c|c|c|c}
                  &~$G_1$,~$G_2$   &~$G_3$          &~$g_j$,~$j\in[\ell]$  \\\hline
           ~$a_1$ &~$s=\siBigAOne, \val=0$ &~$s=0, \val=t$ &~$s=1, \val=w_j$\\
           ~$a_2$ &~$s=0, \val=0$ &~$s=\siBigATwo, \val=t$ &~$s=1, \val=w_j$\\
        \end{tabular}
    \end{center}

    For the correctness, let~$J\subseteq[\ell]$ be a solution for~$\mathcal{I}$.
    We define the bundles as follows.
    We set~$\bundle_1 = \{G_1,G_2\} \cup \{ g_i \mid i \in J \}$ and~$\bundle_2 = \{G_3\} \cup \{ g_i \mid i \in [\ell]\setminus J \}$.
    Clearly, this allocation is SIM.
    Now,~$\val_{a_1}(\bundle_1) = t$, while~$\val_{a_1}(\bundle_2) = 2t$.
    That is, agent~$a_1$ is envious of~$a_2$'s bundle; however, this envy can be removed if~$G_3$ is deleted from~$\bundle_2$.
    For agent~$a_2$, on the other hand, we have that~$\val_{a_2}(\bundle_2) = 2t$, while~$\val_{a_2}(\bundle_1) = 3t$.
    That is, agent~$a_2$ is also envious of~$a_1$.
    S%
    If we remove~$G_1$ from~$a_1$'s bundle, this envy is eliminated.
    That is, the allocation~$\alloc$ is~$\alpha$-SA-EF1 for both agents.

    In the opposite direction, let~$\alloc$ be a MIS allocation.
    It must hold that~$\{G_1,G_2\} \subseteq \bundle_1$ and~$G_3\in \bundle_2$.
    Now, we construct~$J$ so that~$J = \{ j\in[\ell] \mid g_j \in \bundle_1 \}$ and we claim that it is a solution for~$\mathcal{I}$. 
    For the sake of contradiction, assume that~$\sum_{j\in J} w_j \not= \sum_{j\in [\ell]\setminus J} w_j$.
    First, let~$\sum_{j\in J} w_j > \sum_{j\in [\ell]\setminus J} w_j$. 
    Then,~$\val_{a_2}(\bundle_2) = 2t - \delta$ and~$\val_{a_2}(\bundle_1) = 3t + \delta$, where~$\delta = |\sum_{j\in [\ell]\setminus J} w_j - \sum_{j\in J} w_j|$.
    That is, the agent~$a_2$ is envious of~$a_1$, and this envy persists even if we remove~$G_1$ or~$G_2$ from~$\bundle_1$.
    This contradicts that ~$\alloc$ is EF1 for~$a_2$.
    Thus, if also~$\alpha\si_{a_2}(\bundle_2) \le \si_{a_1}(\bundle_2)$, then this cannot be the case.
    This is equivalent to
    \begin{align*}
        \alpha (\siBigATwo + |[\ell] \setminus J|)
        &\geq
        |[\ell] \setminus J| \\
        \alpha
        &\geq
        \frac{|[\ell] \setminus J|}{\siBigATwo + |[\ell] \setminus J|}
    \end{align*}
    and the later holds by the initial setting of~$\siBigATwo$, since the nominator is largest if~$|[\ell] \setminus J| = \ell$ and the denominator if~$|[\ell] \setminus J| = 0$.
    In this case we require that~$\alpha \geq \frac{\ell}{\siBigATwo}$.
    If, on the other hand, we have~$\sum_{j\in J} w_j < \sum_{j\in [\ell]\setminus J} w_j$.
    Then,~$\val_{a_1}(\bundle_1) = t - \delta$ and~$\val_{a_1}(2t + \delta)$.
    Again, this envy persists even if we remove~$G_3$ from~$\bundle_2$.
    This again contradicts that~$\alloc$ is EF1 for~$a_1$.
    Thus, in order for this not to be allowed, we also need that~$\alpha\si_{a_1}(\bundle_1) \geq \si+{a_2}(\bundle_1)$.
    This is equivalent to
    \begin{align*}
        \alpha (2\siBigAOne + |J|)
        &\geq
        |J| \\
        \alpha
        &\geq
        \frac{|J|}{2\siBigAOne + |J|}
    \end{align*}
    and the later holds by the initial setting of~$\siBigAOne$, since the nominator is largest if~$|J| = \ell$ and the denominator if~$|J| = 0$.
    Hence, it must be the case that~$\sum_{j\in J} w_j = \sum_{j\in [\ell]\setminus J} w_j$, which finishes the proof.
\end{proof}

Observe that both \Cref{thm:alphaSA:notGuaranteed} and \Cref{thm:alphaSA:NPc} hold also in the setting where only agent~$a_1$ is not fully socially aware, which paints our results even stronger.

\subsection{Weak Social Awareness}

Although~$\alpha$-SA seems to be a natural relaxation of social awareness, it fails to address one of our most important criticisms. Specifically, the level of social awareness is completely unrelated to the amount of envy between two agents (unless we set~$\alpha = v_i(A_i)/v_j(A_j)$, which is specific for every pair of~$i,j\in\agents$). In order to also take into account this perspective, in the remainder of this section, we explore a novel notion of \emph{weak social awareness} (WSA). Intuitively, the notion formalizes the idea that if an agent~$i$ is envious of the agent~$j$ and the bundle~$\bundle_j$ is twice as good as the bundle~$\bundle_i$ according to the agent~$i$, the social awareness condition should override envy if the social impact generated by~$i$ from~$\bundle_i$ is at least twice the social impact of~$j$ with~$\bundle_i$.
Before we formally define weak social awareness, we provide an example that illustrates how powerful standard social awareness is.

\begin{example}\label{ex:wSA:nonexistence}
    Let us have two agents~$1$ and~$2$, three items~$g_1$,~$g_2$,~$g_3$, and let the valuations and social impacts be as follows.
    \begin{center}
        \renewcommand{\arraystretch}{1.2}
        \begin{tabular}{c|c|c|c}
                  &~$g_1$         &~$g_2$         &~$g_3$  \\\hline
           ~$1$ &~$s=1, \val=1$ &~$s=1, \val=5$ &~$s=0, \val=5$ \\
           ~$2$ &~$s=0, \val=5$ &~$s=1, \val=5$ &~$s=1, \val=1$  \\
        \end{tabular}
    \end{center}
    Let the allocation~$\alloc$ be with~$\bundle_1 = \{g_1\}$ and~$\bundle_2 = \{g_3,g_2\}$. Observe that~$\val_{1}(\bundle_1) = 1$ and~$\val_{1}(\bundle_2) = 10$. That is, the bundle~$\bundle_2$ is ten times more valuable than~$\bundle_1$ according to agent~$1$, and the envy cannot be eliminated by the removal of a single item. However, such an allocation is SA-EF1, since agent~$2$ generates strictly more social impact with~$\bundle_2$. However, the social impact of agent~$2$ is only twice the social impact agent~$1$ would generate with bundle~$\bundle_2$.%
\end{example}

\begin{definition}
\label{def:wsa}
    We say that an allocation~$\alloc$ is \emph{envy-free up to one item with weakly socially aware agents} (WSA-EF1), if for every pair of agents~$i,j\in\agents$, one of the following conditions hold:
    \begin{enumerate}
        \item there is no EF1-envy from~$i$ towards~$j$, or
        \item~$\val_i(\bundle_j)\cdot \si_i(\bundle_j) \leq \val_i(\bundle_i)\cdot s_j(\bundle_j)$.
    \end{enumerate}
\end{definition}

The WSA variants for other fairness notions are then defined analogously, with appropriate modifications to condition 1.

\begin{example}\label{ex:wSA-EF1}
    Recall the instance from \Cref{ex:wSA:nonexistence}. Clearly, the presented allocation~$\alloc$ is not EF1 from the perspective of agent~$1$. At the same time, we have~$\val_{1}(\bundle_2)\cdot \si_1(\bundle_1) = 10 \cdot 1 = 10$ and~$\val_{1}(\bundle_1)\cdot \si_2(\bundle_2) = 1\cdot 2 = 2$. Therefore, neither condition 2 from the definition of WSA-EF1 is satisfied, which means that~$\alloc$ is not a WSA-EF1 allocation. 
\end{example}

It is easy to see that the instance from \Cref{ex:wSA:nonexistence} admits no SIM and WSA-EF1 allocation, as only other possible allocation is with~$\bundle_1 = \{g_1\}$ and~$\bundle_2 = \{g_2,g_3\}$, which is also not WSA-EF1 by symmetric arguments. %
Moreover, also with weakly SA agents, the problem remains intractable. 

\begin{theorem}
\label{thm:weaklySA:NPh}
    For any~$\mathcal{F}\in\{\text{tEF1},\text{sEF1},\text{wEF1},\text{swEF1},\text{EF1},\text{EFL}\}$, it is \NPc to decide whether SIM and WSA-$\mathcal{F}$-fair allocation exist, even if the social impact function is binary and~$\numAgents = 2$.
\end{theorem}
\begin{proof}
    This time, we reduce from a variant of the \probName{Partition} problem called \probName{Equitable Partition}. In this variant, we are given a multiset~$W$ of~$2\ell$ integers whose sum is~$2t$, and the goal is to decide whether~$J\subseteq[2\ell]$ of size~$\ell$ exists such that~$\sum_{j\in J} w_j = \sum_{j\in[2\ell]\setminus J} w_j = t$. This problem is known to be \NPh, even if for each~$J\subseteq[2\ell]$ of size less than~$\ell$ we have~$\sum_{j\in J} w_j < t$~\cite{DeligkasEKS2024}. Without loss of generality, we can assume that~$\ell > 4$, as otherwise, we can solve such instance in polynomial by a simple brute-force.

    Let~$\mathcal{I}$ be an instance of \probName{Equitable Partition}. We construct an equivalent instance of \FDSI[WSA-EF1] as follows (again, we discuss the other notion of fairness at the end of the proof). There are two agents,~$a_1$ and~$a_2$. The set of items contains two \emph{large items}~$G_1$ and~$G_2$ together with~$2\ell$ \emph{small items}~$g_j$,~$j\in[2\ell]$, which are in one-to-one correspondence with integers of~$W$. The valuations and the social impact functions are defined according to the following table.
    
    \begin{center}
        \renewcommand{\arraystretch}{1.2}
        \begin{tabular}{c|c|c|c}
                  &~$G_1$         &~$G_2$         &~$g_j$,~$j\in[2\ell]$  \\\hline
           ~$a_1$ &~$s=1, \val=0$ &~$s=0, \val=t$ &~$s=1, \val=w_j$\\
           ~$a_2$ &~$s=0, \val=t$ &~$s=1, \val=0$ &~$s=1, \val=w_j$\\
        \end{tabular}
    \end{center}

    For correctness, let~$\mathcal{I}$ be a \Yes-instance, and let~$J$ be a solution. We set~$\bundle_{a_1} = \{G_1\} \cup \{ g_j \mid j\in J \}$ and~$\bundle_{a_2} = \items \setminus \bundle_{a_1}$. Clearly, the allocation is SIM, since~$G_1$ is allocated to~$a_1$ and~$G_2$ is allocated to~$a_2$. Now, we claim that~$\alloc$ is WSA-EF1. For~$a_1$, we have~$\val_{a_1}(\bundle_{a_2}) = t$ and~$\val_{a_1}(\bundle_{a_2}) = 2t$; however, this envy can be eliminated by removing~$G_2$ from~$\bundle_{a_2}$. The situation for~$a_2$ is symmetric: we have~$\val_{a_2}(\bundle_{a_2}) = t$ and~$\val_{a_2}(\bundle_{a_1}) = 2t$ and envy can be eliminated by removing~$G_1$ from~$\bundle_{a_1}$.

    In the opposite direction, let~$\alloc=(\bundle_{a_1},\bundle_{a_2})$ be a SIM and WSA-EF1 allocation for~$\mathcal{J}$. Since~$\alloc$ is SIM, it must be the case that~$G_1\in\bundle_{a_1}$ and~$G_2\in\bundle_{a_2}$. For the sake of contradiction, assume that~$|\bundle_{a_1}| \not= |\bundle_{a_2}|$, and, without loss of generality, let~$|\bundle_{a_1}| < |\bundle_{a_2}|$. Then~$\bundle_{a_1}$ contains at most~$\ell-1$ small items. Consequently,~$\val_{a_1}(\bundle_{a_1}) = t - \delta$ and~$\val_{a_1}(\bundle_{a_2}) = t + t + \delta = 2t + \delta$ and this envy cannot be eliminated even if we remove~$G_2$ from~$\bundle_{a_2}$. Hence, as~$\alloc$ is WSA-EF1, it must be the case that
    \[
        \val_{a_1}(\bundle_{a_2})\cdot \si_{a_1}(\bundle_{a_2}) \leq \val_{a_1}(\bundle_{a_1})\cdot \si_{a_2}(\bundle_{a_2})\,.
    \]
    However, if we substitute the actual values, we obtain 
    \begin{align*}
        (2t+\delta)\cdot(|\bundle_{a_2}|-1) &\leq (t-\delta)(|\bundle_{a_2}|)\\
        2t|\bundle_{a_2}| - 2t + \delta|\bundle_{a_2}| - \delta &\leq t|\bundle_{a_2}| - \delta|\bundle_{a_2}|\\
        t|\bundle_{a_2}| - 2t + 2\delta|\bundle_{a_2}| - \delta &\leq 0\,,
    \end{align*}
    which is clearly never satisfied, as~$|\bundle_{a_2}| > \ell + 1$ and~$\ell > 4$ by our assumptions. Therefore,~$\alloc$ does satisfy neither the second condition of WSA-EF1, which is a contradiction. Hence, it must be the case that~$|\bundle_{a_1}| \geq |\bundle_{a_2}|$. By symmetric argument, we have that also~$|\bundle_{a_2}| \geq |\bundle_{a_1}|$, meaning that~$|\bundle_{a_1}| = |\bundle_{a_2}| = \ell$. We create~$J$ such that~$J = \{ j\in[2\ell] \mid g_j\in \bundle_{a_1} \}$. By the previous arguments, we have~$|J| = \ell$. It remains to show that~$\sum_{j\in J} w_j = t$. Without loss of generality, let~$\sum_{j\in J} w_j < t$. Then,~$\val_{a_1}(\bundle_{a_2}) = 2t + \delta$ and~$\val_{a_1}(\bundle_{a_1}) = t - \delta$. However, we have already shown that in this situation,~$\alloc$ is not WSA-EF1, so we have a contradiction. Therefore, it must be the case that~$\sum_{j\in J} w_j = t$, which finishes the proof.
\end{proof}

\section{Discussion}\label{sec:discussion}

We have presented a comprehensive study of computing allocations that maximize social impact while maintaining some fairness criterion. 
We believe that our results resolve almost completely the problem  for goods when the fairness criterion is based on EF1, or EFL. 
It is evident that in order to achieve SIM, social awareness is essential.
Additionally, our results from \Cref{sec:relaxedSA} indicate that probably the definition of social awareness from~\cite{FlamminiGV2025}, i.e., the one we have adopted, is probably the ``correct'' one, as both relaxations we have identified remove the existential guarantee and furthermore make the problem hard. Hence, an immediate question is whether there exist other natural definitions of social awareness that allow for SIM and fair allocations. In the following, we highlight three other directions that we believe deserve further study.

\paragraph{Relaxing SIM.}
One obvious direction for future research is to relax the SIM constraint, as someone can argue that SIM is too restrictive and {\em this} is the reason for non-existence of fair allocations for socially unaware agents; \citet{FlamminiGV2025} show that we cannot get better than~$\Oh{\numAgents}$-approximation for SIM under EF1 fairness.
However, what if we replace SIM with a Pareto optimal allocation with respect to social impact? Unfortunately, there is no positive news for this: the proof of \Cref{thm:two_agents_hard} shows that the problem remains \NPc, even when there are only two agents.

\paragraph{Other Fairness Notions.}
A complementary approach to the question above is to further relax the fairness criterion. Does there always exist an allocation that is SIM and {\em approximately} fair, for some established fairness notion?
A different idea is to try to strengthen the fairness guarantee for socially aware agents. Can we prove the existence of SIM and EFX allocations with socially aware agents for the settings that admit EFX allocations in the ``standard'' model? What about SIM allocations satisfying epistemic EFX, that are guaranteed to exist and can be efficiently found~\cite{AzizBCGL2018,Caragiannis_epistemic,AkramiR2025}?

\paragraph{Division of Chores.}Another exciting avenue, which we believe deserves further study, is the setting where we have {\em chores} to allocate~\cite{Aziz2016,AzizRSW2017,AzizCIW2022}. 
This is a very natural setting with many practical applications; think of the admin tasks in a CS department. 
Our initial results provided in the appendix show that, again, social unawareness does not allow for SIM and EF1 allocations.
Besides this, even with social awareness, it is unclear whether we can always find SIM-EF1 allocations. In fact, we have created examples where round-robin approach, that gives us SA-EF1 for goods by \Cref{thm:sa_swef1_goods}, fails when we have chores.

\begin{example}
    Consider an instance with two agents with identical valuation~$v$ and five items. For the first two items, both agents have value~$v(g_1)=v(g_2)= -100$ and~$s_i(g_1)=s_i(g_2)= 1$ for~$i\in \{1,2\}$. For the remaining three items~$v(g_3)=v(g_4)= v(g_5) = -1$ but~$s_1(g_j) = 0$ and~$s_2(g_j)=1$ for~$j\in \{3,4,5\}$. Clearly, if we allocate the chores in round-robin fashion such that in each round agents pick their best chore for which they maximize the social impact, if such chore exists, then we end up with allocation~$(\{g_1, g_2\}, \{g_3, g_4, g_5\})$. However, since~$s_1(\{g_1, g_2\}) = s_2(\{g_1, g_2\})$ and~$v(\{g_1\}) = v(\{g_2\}) = -100 < v(\{g_3, g_4, g_5\}) = -3$, agent 1 envies agent 2. 
\end{example}
\noindent

Similarly, it is unclear whether algorithms based on envy-cycle elimination can be extended to SA-EF1 for chores. The main issue is that in the setting of chores, it is possible that agent~$i$ does not SA-envy bundle~$\bundle_j$, when agent~$j$ has the bundle, as agent~$j$ has a smaller impact on the bundle of agent~$i$; however, if~$\bundle_j$ is moved to some other agent~$k$, who has same social impact for the~$i$-th's bundle as~$i$, the agent~$i$ can start SA-envying the same bundle. This is a problem, as obtaining SIM allocation requires eliminating the cycle only in the SA-envy graph. However, this means that it is possible that the agent to whom we allocate a chore already envies some bundle, just not SA-envies it on the current agent. Moving this bundle then can introduce envy that is not SA-EF1. 
Hence, the discussion above is a strong indication, at least to us, that if SA-EF1 allocations do indeed exist for chores, they require a different algorithmic technique in order to resolve it.

\begin{acks}
    This project has received funding from the European Research Council (ERC) under the European Union’s Horizon 2020 research and innovation programme (grant agreement No 101002854). This work was co-funded by the European Union under the project Robotics and advanced industrial production (reg. no. CZ.02.01.01/00/22\_008/0004590). Argyrios Deligkas acknowledges the support of the EPSRC grant EP/X039862/1.
    Šimon Schierreich acknowledges the additional support of the Grant Agency of the Czech Technical University in Prague, grant No. SGS23/205/OHK3/3T/18.
    
    \begin{center}
    	\includegraphics[width=4cm]{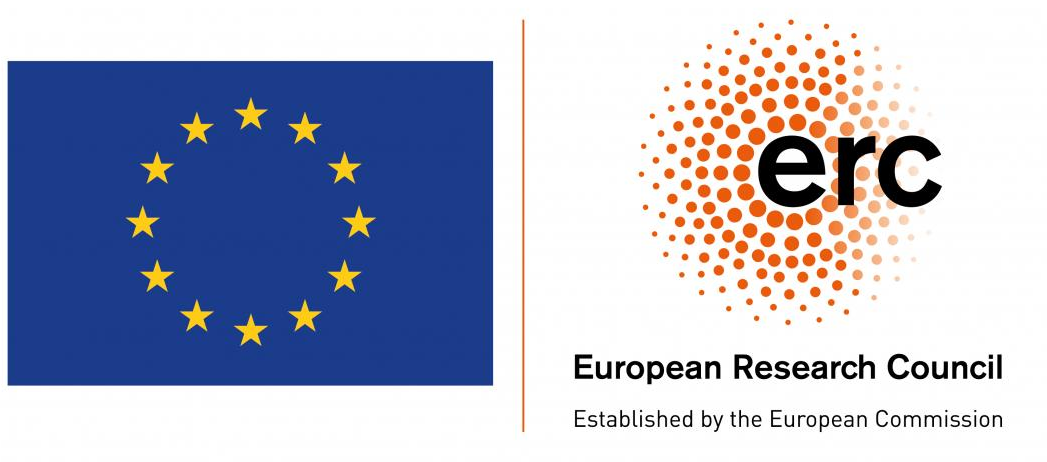}
    \end{center}
\end{acks}

\printbibliography

\appendix

\end{document}